\RequirePackage{fix-cm}
\documentclass[smallextended, final]{svjour3}       
\smartqed  
\usepackage{graphicx}
\usepackage{amsmath,amssymb}
\usepackage{bbm,multirow,nicefrac,xcolor,tikz,subfig,afterpage,pdflscape,geometry}
\newtheorem{fact}{Fact}

\newcommand\Osq{\mathbin{\text{\scalebox{.84}{$\square$}}}}
\newcommand{\paren}[1]{\left(#1\right)}
\newcommand{\set}[1]{\left\{#1\right\}}
\newcommand{\Z}{\mathbb{Z}}
\newcommand{\R}{\mathbb{R}}
\newcommand{\C}{\mathbb{C}}
\renewcommand{\vec}[1]{\mathbf{#1}}
\newcommand{\CC}{\mathrm{CC}}
\DeclareMathOperator{\BW}{BW}

\newcommand{\one}{\mathbbm{1}}
\DeclareMathOperator{\Aut}{Aut}
\DeclareMathOperator{\SPEC}{spec}
\DeclareMathOperator{\SPAN}{span}
\DeclareMathOperator{\CCC}{CCC}
\DeclareMathOperator{\DV}{DataVortex}
\DeclareMathOperator{\PT}{PT}
\DeclareMathOperator{\DF}{DragonFly}
\DeclareMathOperator{\BF}{Butterfly}
\DeclareMathOperator{\diam}{diam}
\DeclareMathOperator{\SF}{SlimFly}

\newcommand{\conn}{\textrm{\reflectbox{$\rightsquigarrow$}}}
\newcommand{\spec}[1]{\SPEC\!\paren{#1}}
\newcommand{\linspan}[1]{\SPAN\!\set{#1}}
\newcommand{\inner}[3][]{\ensuremath{\left< #2 , #3 \right>_{#1}}}
\newcommand{\norm}[2][]{\ensuremath{\left\| #2 \right\|_{#1}}}
\newcommand{\size}[1]{\left| #1 \right|}
\newcommand{\order}[1]{\left\| #1 \right\|}
\newcommand{\hidden}[1]{}
\newcommand{\bigOh}[1]{\mathcal{O}\!\paren{#1}}
\newcommand{\lilOh}[1]{\ensuremath{\it{o}\!\paren{#1}}}
\newcommand{\floor}[1]{\left\lfloor #1 \right\rfloor}
\newcommand{\ceil}[1]{\left\lceil #1 \right\rceil}

\newcommand{\abs}[1]{\left|#1\right|}

\begin{document}

\title{Ramanujan Graphs and the Spectral Gap of Supercomputing Topologies
}


\author{Sinan G. Aksoy  \and Paul Bruillard \and Stephen~J.~Young \and
Mark Raugas}


\institute{Sinan G. Aksoy\at
	Pacific Northwest National Laboratory \\
	Richland, WA 99352 \\
              \email{sinan.aksoy@pnnl.gov}           
           \and
           Paul Bruillard \at
              \email{bruillardp@gmail.com}           
           \and
           Stephen J. Young \at
	Pacific Northwest National Laboratory \\
	Richland, WA 99352 \\
              \email{stephen.young@pnnl.gov}           
              \and
              Mark Raugas \at
              Pacific Northwest National Laboratory \\
              Seattle, WA 98109 \\
              \email{mark.raugas@pnnl.gov}
}

\date{\today} 

\maketitle

\begin{abstract}
Graph eigenvalues play a fundamental role in controlling structural
properties which are critical considerations in the design of
supercomputing interconnection networks, such as bisection bandwidth, diameter, and fault
tolerance. This motivates considering graphs with optimal spectral expansion, called {\it Ramanujan graphs}, as potential candidates for interconnection networks. In this work, we explore this possibility by comparing Ramanujan graph properties against those of a wide swath of current and proposed supercomputing topologies. We derive analytic expressions for the spectral gap, bisection bandwidth, and diameter of these topologies, some of which were previously unknown. We find the spectral gap of existing topologies are well-separated from the optimal achievable by Ramanujan topologies, suggesting the potential utility of adopting Ramanujan graphs as interconnection networks. \\
\end{abstract}

\section{Introduction}



One of the significant challenges in the use of modern cluster-based
supercomputers is efficiently, robustly, and quickly handling the necessary communication between nodes in the cluster. Both the current and next-generation supercomputer designs use highly structured network topologies, such as the low-dimensional torus, the flattened butterfly, or the dragonfly topology in order to have a straightforward routing scheme while attempting to mitigate the traffic congestion in high communication applications. However, ``preliminary experiments on Edison, a Cray XC30 at {NERSC}, have shown that for communications-heavy applications, inter-job interference and thus network congestion remains an important factor" \cite{Bhatele:Dragonfly}. {Indeed, recent research \cite{Prieto-Castrillo2014} further attests to the impact of network structure on performance metrics.} In fact, even with a relatively low utilization (40-50\%), communication patterns can cause an exponential explosion in latency \cite{Kim2008}. As a consequence of the interaction between the structure of internode communication in various classes of algorithms and the underlying network topologies, certain supercomputers gain a reputation for being more or less suited to a certain class of problems.

In this regard, the evolution of supercomputing interconnection topologies stands in contrast to the surprising success of the ``evolved" topology of the Internet. Specifically, despite having no global design, the Internet structure has unexpectedly \cite{InternetCollapse} ended up as a robust, general purpose, and relatively low-latency system for its size. In the last few decades, a consensus has developed that the primary explanation for the good performance of the internet topology is that the internet topology belongs to a class of graphs known as {\it expanders}. That is, if a graph is a sufficiently high-quality expander then there exists efficient, distributed, online, local, and low-congestion algorithms to route information among the vertices of the graph \cite{Chung:RoutingPermutations,Chung:spectral,Frieze:DisjointExpander,Mihail:congestion,Kleinberg:ShortExpander,Vazirani:approx}. 


This view point leads naturally to considering optimal expanders,
known as Ramanujan graphs, as potential supercomputing topologies.  In
this work, we explore the potential benefits of adopting Ramanujan
graphs by conducting an analysis of current and proposed
supercomputing topologies. The paper is organized as follows. In
Section \ref{sec:back}, we provide the necessary preliminaries on
spectral graph theory, as well as survey results showing eigenvalues control a number of critical properties pertinent to interconnection design, such as bisection bandwidth, diameter, and fault tolerance. Second, in Section \ref{sec:raman} we define the Ramanujan property of graphs, and review explicit constructions of Ramanujan graphs. In Section \ref{sec:spec}, we survey variety of supercomputing topologies and derive analytic expressions for their spectral expansion, bisection bandwidth, and diameter. Across the topologies surveyed, we find some or all of these properties are well separated from those of Ramanujan topologies. Consequently, our results suggest transition to Ramanujan topologies may have the potential to significantly improve metrics for facility of communication.

%
%
%
%
%
%
%
%
%
%


\section{Preliminaries}\label{sec:back}

Before proceeding with our discussion of expanders and Ramanujan graphs, we first recall some relevant terminology and results from graph theory. A graph $G=(V,E)$ is a set of vertices  $V$ edges $E$, where each edge is an unordered pair of vertices. The number of edges incident to a vertex is called its {\it degree}; if every vertex has degree $k$, the graph is called {\it $k$-regular}. Spectral graph theory is the study of eigenvalues and eigenvectors of matrices associated with graphs. The {\it adjacency matrix} $A$ of an $n$-vertex graph is an $n \times n$ matrix where
\[
A_{ij}=\begin{cases} 1 & \mbox{ if } \{i, j\} \in E \\ 0 & \mbox{ otherwise} \end{cases}.
\]
As $A$ is symmetric, its eigenvalues are real, which we denote
\[
\lambda_1 \geq \lambda_2 \geq \dots \geq \lambda_n.
\] 
For a connected graph, the largest eigenvalue $\lambda_1=k$ if and only if the graph is $k$-regular; furthermore, if $G$ is connected, $\lambda_1 -\lambda_2 > 0$, and the quantity $\lambda_1-\lambda_2$ is referred to as the {\it spectral gap} of $G$. 
\begin{figure}
\centering
\includegraphics[scale=0.34]{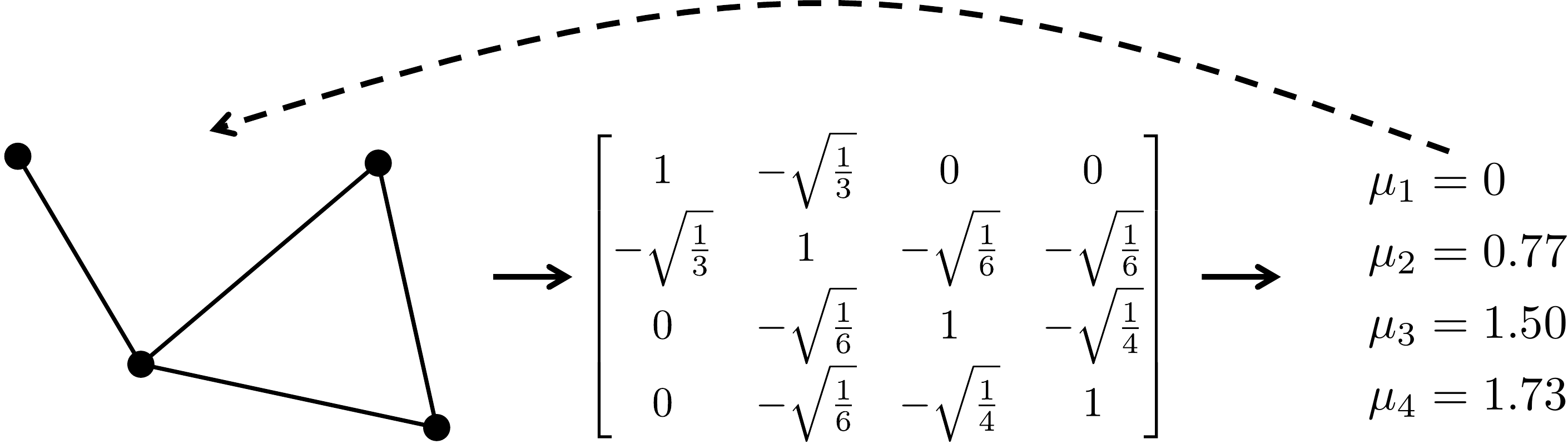}
\caption{{How eigenvalues are obtained from a graph, represented by the solid lines. Spectral graph theory analyzes eigenvalues to deduce properties of the graph, represented by the dashed line. }}\label{specTh}
\end{figure}

Two other graph matrices that whose spectra is often studied are the {\it Laplacian matrix} $L =D-A$ and {\it normalized Laplacian matrix} $\mathcal{L}=D^{-1/2}LD^{-1/2} $, where $D$ denotes the diagonal matrix with the vertex degrees on the diagonal. Unlike the adjacency matrix, both of these matrices are necessarily positive semi-definite, and their spectra characterizes a number of properties which are not captured by adjacency eigenvalues. Due to its intimate connection to random walks and stochastic processes on graphs, the normalized Laplacian matrix is perhaps the most appropriate matrix for characterizing expansion properties of graphs, particularly for irregular graphs. However, we note that if a graph $G$ is $k$-regular (as is the case for a number of supercomputing topologies), then $\mathcal{L}=I-\tfrac{1}{k}A$, from which it is clear that the spectra of all three matrices are related by trivial shifts and scalings by $k$ (and hence functionally the same). We denote the spectrum of the Laplacian matrix $L$ by
\[
0=\rho_1 \leq \rho_2 \leq \dots \leq \rho_n,
\]
and that of the normalized Laplacian $\mathcal{L}$ by
 \[
 0=\mu_1 \leq \mu_2 \leq \dots \leq \mu_n \leq 2.
 \]
 {We give an example of a graph, its normalized Laplacian, and associated eigenvalues in Figure \ref{specTh}}. As we will later see, the eigenvalues $\lambda_2, \rho_2$, and $\mu_2$ play a critical role in controlling expansion properties and defining Ramanujan graphs. In particular, the eigenvalue $\rho_2$ is called the {\it algebraic connectivity} of a graph. Due to its prevalence in the literature (see for instance \cite{Brouwer:SpectraGraphs,Biggs:AGT,Godsil:AGT}), we will choose to present our results in terms of this spectrum, keeping in mind that if $G$ is $k$-regular, then
\[
\rho_2=k\cdot \mu_2 = k - \lambda_2.
\]

Before proceeding, we describe the spectra of two graphs: the path and the cycle graph. We highlight these graphs as they are frequently elemental to the design of fundamental topologies (e.g. the torus, mesh, and hypercube are all obtained via graph products of cycles or paths). Unsurprisingly, their spectra is highly structured. 

\begin{itemize}

\item The path of length $n-1$, denoted $P_{n}$, has $n-1$ edges and $n$ vertices, and adjacency spectrum
\[
2\cos\left(\frac{\pi j }{n+1}\right) \mbox{ for } j \in \set{1,\dots,n}.
\]
\item If the path of length $n-1$ is modified to add self-loops at each of the endpoints, denoted $P'_{n}$, the adjacency spectrum becomes 
\[
2\cos\left(\frac{\pi j }{n}\right) \mbox{ for } j \in \set{0,\dots,n-1}.
\]
\item The cycle of length $n$, denoted $C_n$, has $n$ edges and vertices, and adjacency spectrum
\[
2\cos\left(\frac{2 \pi j}{n}\right) \mbox{ for } j \in \set{0,\dots,n-1}. 
\]
\end{itemize}

Finally, we use standard asymptotic notation: a function $f(n)=O(g(n))$ if for all sufficiently large values of $n$ there exists a positive constant $c$ such that $|f(n)| \leq c \cdot |g(n)|$; similarly, we write $f(n)=\Omega(g(n))$ if $g(n) = O(f(n))$, and $f(n)=\Theta(g(n))$ if both $f(n)=O(g(n))$ and $f(n)=\Omega(g(n))$. Lastly, $f(n)=o(g(n))$ if $\lim_{n \to \infty} \frac{f(n)}{g(n)}=0$.

\subsection{Network Properties}
Graph eigenvalues are deeply related to a number of fundamental network properties. In the case of supercomputing topologies, two such properties linked to communications performance are graph {\it diameter} and {\it bisection bandwidth}. Diameter (the maximum distance between vertices) is critical for latency, while bisection bandwidth (the minimum number of edges crossing a balanced bipartition of the vertices) measures the networks ``bottleneckedness", impacting all-to-all communication performance.

A plethora of work has shown both of these core network properties to be bounded and thus controlled by graph eigenvalues \cite{Chung1989}. In particular, the eigenvalues of interest are the {\it spectral gap}, the difference in the largest two adjacency eigenvalues, or the {\it algebraic connectivity}, the second smallest Laplacian eigenvalue. For example, Alon and Milman \cite{Alon1985} showed that the diameter is at most roughly $C \cdot \log{n}$, where $C$ depends on algebraic connectivity and the maximum degree. More precisely: 

\begin{theorem}[Alon, Milman 1985] \
Let $G$ be an $n$-vertex graph with algebraic connectivity $\rho_2$ and maximum degree $\Delta$. Then
\[
\diam(G)\leq 2 \left \lceil \sqrt{\frac{2\Delta}{\rho_2}} \log_2{n} \right\rceil
\]
\end{theorem}
A lower bound on graph diameter in terms of algebraic connectivity may also be obtained. For example, McKay \cite{Mohar:Eigenvalues} showed  $\mbox{diam}(G)\geq \tfrac{4}{n\rho_2}$. In addition to these bounds on the maximum distance between vertices, {\it average} distance is upper and lower bounded in terms of algebraic connectivity as well; see \cite{Mohar:Eigenvalues}. Next, algebraic connectivity provides guarantees on minimum bisection bandwidth, as shown by Fiedler \cite{fiedler1973algebraic}.

\begin{theorem}[Fiedler 1975] \label{T:lower_bw}
Let $G$ be an $n$-vertex graph with algebraic connectivity $\rho_2$ and bisection bandwidth $\mbox{BW}(G)$. Then
\[
\BW(G)\geq \frac{\rho_2 n}{4}.
\]
\end{theorem}

By considering Cheeger's inequality \cite{Sokal:gap,Jerrum:gap} one can also obtain upper bounds on the bisection bandwidth in terms of $\rho_2$ for regular graphs. 
\begin{theorem}
For a connected $k$-regular, $n$-vertex graph $G$ with algebraic connectivity $\rho_2$, the bisection bandwidth satisfies
\[
\BW(G) \leq \frac{\sqrt{2k\rho_2} \cdot kn}{2}.
\]
\end{theorem}

We note that when $\rho_2$ is large this upper bound is quite loose.
In fact, if $G$ has $m$ edges, an easy application of the first moment
method~\cite{Alon:prob_meth} shows the bisection bandwidth is at most $\frac{m}{2}$.  Note
that if $G$ is $k$-regular and $\rho_2$ is asymptotically $k$, then
this first moment calculation shows that Theorem
\ref{T:lower_bw} is essentially tight and the bisection bandwidth is
$\frac{kn}{4}\paren{1+\lilOh{1}}$. Consequently, it can be shown that
Ramanujan graphs (defined in Section \ref{sec:raman}) have nearly
optimal bisection bandwidth among all $k$-regular graphs.

Lastly, we note that algebraic connectivity provides bounds on edge and vertex {\it connectivity}, the minimum number of edges and vertices that must be deleted in order to disconnect the graph, respectively. In the context of computer interconnection networks, vertex connectivity is often referred to as {\it fault tolerance} (e.g., \cite{Akers1987a}); more precisely, fault tolerance is defined as one less than vertex connectivity. Denoting vertex and edge connectivity as $\kappa(G), \kappa'(G)$ respectively, it is obvious that $\kappa(G) \leq \kappa'(G) \leq \Delta(G)$. Fiedler \cite{fiedler1973algebraic} proved
\[
\kappa(G)\geq \rho_2,
\]
hence, larger algebraic connectivity guarantee more robust fault tolerance. For more spectral bounds on vertex and edge connectivity, the reader is referred to \cite{Abiad2018} and for further a more complete survey of the relationship between algebraic connectivity and numerous graph invariants, see \cite{mohar1991laplacian}. Such spectral bounds have practical utility: for a number of graph topologies, exact diameter, bisection bandwidth, etc, may be unknown or difficult to compute and hence eigenvalues may serve as a proxy. In summary, the bounds we've reviewed motivate algebraic connectivity as a key parameter of interest intimately related to a plethora of structural properties important to interconnection network design. In the next section, we define graphs with optimal spectral gap, known as {\it Ramanujan graphs}, and discuss their expansion properties. 

%
%




\section{Ramanujan Graphs} \label{sec:raman}

\begin{figure}
\centering
\begin{minipage}{.4\textwidth}
  \centering
\includegraphics[width=.5\linewidth]{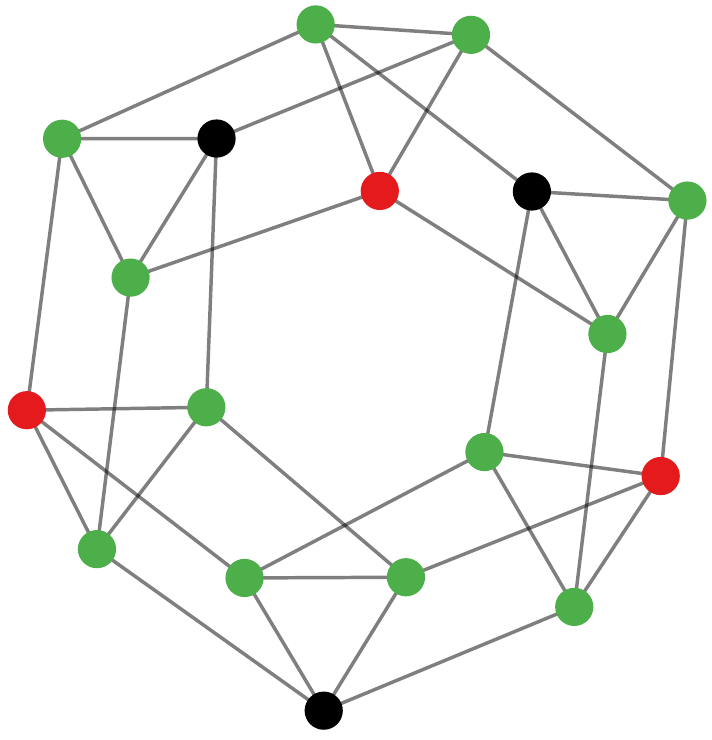}
  \captionof{figure}*{{\bf (a)}}
  \label{fig:test1}
\end{minipage}%
\begin{minipage}{.4\textwidth}
  \centering
\includegraphics[width=0.5\linewidth]{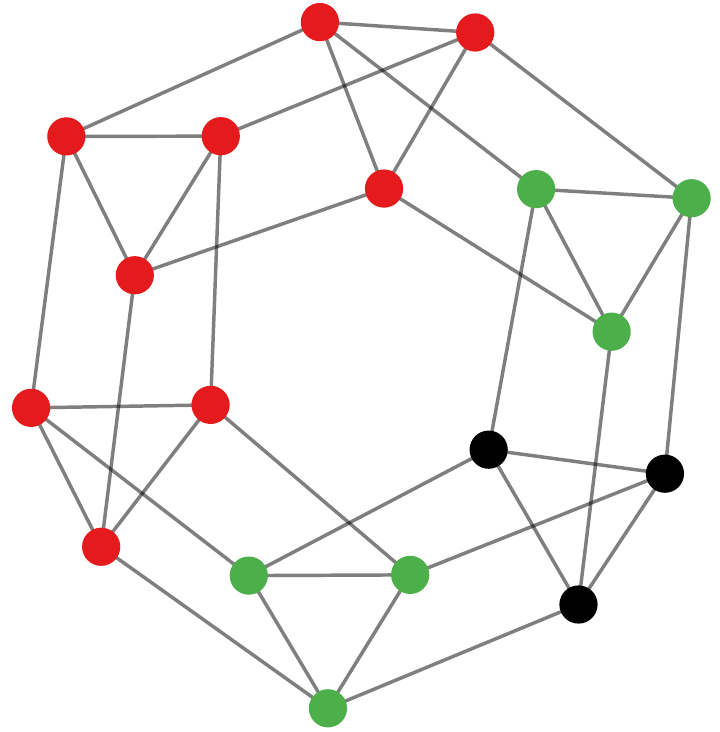}
  \captionof{figure}*{{\bf (b)}}
  \label{fig:test2}
\end{minipage}
\caption{{Examples of a vertex set $X$, in red, with vertex boundary $\partial X$, in green. For (a), the set has vertex isoperimetric ratio 4. For (b), this ratio is $\nicefrac{2}{3}$, which is the minimum over all subsets whose size doesn't exceed half the vertices, and thus is the graph's vertex isoperimetric number.}}\label{iso_example}
\end{figure}

Ramanujan graphs are regular graphs with nearly optimal {expansion} properties. Loosely speaking, expansion means that every ``not too large" set of vertices has a ``not to small" set of neighbors. One way of measuring such expansion is the {\it vertex isoperimetric number} of a graph, given by
\[
h(G)=\min_{\substack{X \subseteq V(G) \\ 2|X|\leq |V(G)|}}\frac{|\partial X|}{|X|},
\]
where $\partial X$ denotes the neighbors of vertices in $X$ that are not in $X$. {We illustrate examples of vertex isoperimetric ratios of sets in Figure \ref{iso_example}.} This notion of expansion, as well as others such as the {\it edge isoperimetric constant}, have been shown to be intimately related to the second largest adjacency eigenvalue of a graph. For example, Tanner \cite{Tanner1984} proved a lower bound on $h(G)$ in terms of this eigenvalue $\lambda_2$, for a $k$-regular graph; namely,
\[
h(G) \geq 1- \frac{k}{2k-2\lambda_2}.
\]
Conversely, Alon and Milman \cite{Alon1985} proved an upper bound on $\lambda_2$ in terms of $h(G)$:
\[
k-\lambda_2 \geq \frac{h(G)^2}{4+2\cdot h(G)^2}.
\]
Putting these two bounds together, it is clear that smaller values of $\lambda_2$ yield larger values of $h(G)$ and hence better expansion. Other bounds, such as Cheeger's inequality and Buser's inequality \cite{Buser1982}, similarly tie eigenvalues to other notions of expansion, like the Cheeger constant. Given the breadth of expansion properties reflected through eigenvalues, it is natural to measure expansion directly in terms of the spectra itself. Accordingly, researchers have sought {\it spectral expanders}, families of graphs with small $\lambda_2$. The most well-known such family are called {\it Ramanujan graphs}.

\begin{definition} 
A $k$-regular graph $G$ is called Ramanujan if 
\[
\lambda(G) \leq 2 \sqrt{k-1},
\]
where $\lambda(G)$ denotes the largest magnitude adjacency eigenvalue of $G$ not equal to $\pm k$. 
\end{definition}

Ramanujan graphs are, in a sense, optimal spectral expanders since they achieve the asymptotic theoretical minimum given by Alon-Boppana theorems. The Alon-Boppana theorem \cite{Alon:EigenvaluesExpanders,Nilli:AlonBoppana} states that for a $k$-regular graph with second largest (in magnitude) adjacency eigenvalue $\lambda$ and diameter $D$, we have
\[
\lambda \geq 2\sqrt{k-1}\left(1-\frac{2}{D}\right)-\frac{2}{D}.
\]
As an immediate corollary, if $(G_i)_{i=1}^{\infty}$ is a family of connected, $k$-regular, $n$-vertex graphs with $n \to \infty$ as $i \to \infty$, then,
\[
\liminf_{i \to \infty} \lambda(G_i) \geq 2 \sqrt{k-1}.
\] 
Hence, we see that Ramanujan graphs attain the theoretical asymptotic optimum spectral expansion. While the Alon-Boppana theorem pertains to regular graphs, variants of the theorem have been proposed for the case of irregular graphs, see \cite{chung2016generalized,Hoory:AlonBoppana,Young:AlonBoppana}.

As a consequence of their optimal spectral expansion, Ramanujan graphs possess beneficial structural properties via the bounds mentioned in Section \ref{sec:back}.  In particular, not only does the Ramanujan property guarantee at least nearly optimal bisection bandwidth, but also controls the number of edges between {\it any} collection of vertices, not just bisections. This stronger property is known as the discrepancy property~\cite{Chung:spectral}.  Specifically, using tools of spectral graph theory, if $G$ is an $n$-vertex $k$-regular Ramanujan graph we have that for any two sets of vertices $X$ and $Y$, 
\[ \size{ e(X,Y) - \frac{k}{n}\size{X}\size{Y}} \leq \frac{2\sqrt{k-1}}{n}\sqrt{\size{X}(n-\size{X})\size{Y}(n-\size{Y})}, \] 
where $e(X,Y)$ is the number of edges between the sets $X$ and $Y$.
Roughly speaking, this says that in any Ramanujan topology the number
of edges between two sets scales roughly like the expected number of
edges between two sets in a similarly dense random graph.  In particular, if a process is active on $\alpha n$ fraction of the nodes of the supercomputing topology, then bisection bandwidth on the active nodes is at least \[ \frac{\alpha kn}{2} \paren{ \frac{\alpha}{2} - \frac{2\sqrt{k-1}}{k} \paren{1-\frac{\alpha}{2}}} \] independently of which $\alpha n$ nodes are chosen.

\subsection{Ramanujan Constructions}


Providing explicit constructions of Ramanujan graphs is challenging. The first explicit constructions of Ramanujan graphs were given by Lubotzky, Phillips, and Sarnak \cite{Lubotzky1988}, as well as independently by Margulis \cite{margulis1988explicit}. Both constructions are Cayley graphs that rely heavily on number-theoretic methods; indeed, the name ``Ramanujan graph" was derived due to the application of the Ramanujan--Petersson conjecture from number theory in the aforementioned construction \cite{Lubotzky1988}. Below, we briefly describe and compare some of these constructions.

\subsubsection{Lubotzky, Phillips, Sarnak Construction}

\begin{definition}[LPS Graphs]
The LPS graph $X^{p,q}$ is a $(q+1)$-regular Cayley graph, defined for distinct primes $p$ and $q$ such that $p,q \equiv 1 \pmod{4}$. Letting $i$ be any integer such that $i^2 \equiv -1 \pmod{p}$, the generating set $S$ of $X^{p,q}$ is given by
\[
S=\left\{ \begin{bmatrix} \alpha_0+i\alpha_1 & \alpha_2 + i \alpha_3 \\
-\alpha_2 + i \alpha_3 & \alpha_0-i\alpha_1
 \end{bmatrix} \  \middle \vert \begin{array}{l} (\alpha_0,\alpha_1,\alpha_2,\alpha_3) \mbox{ is a solution of } \alpha_0^2 + \alpha_1^2 + \alpha_2^2+\alpha_3^2 = q, \\
 \alpha_0>1 \mbox{ is odd, and } \alpha_1,\alpha_2,\alpha_3 \mbox{ are even.} \end{array}  \right\},
\]
and the group $G$ of $X^{p,q}$ is
\[
G=
\begin{cases}
\mathrm{PSL}(2,\mathbb{F}_p) & \mbox{ if } \left(\frac{q}{p}\right)=1 \vspace{1mm}  \\ 

\mathrm{PGL}(2,\mathbb{F}_p) & \mbox{ if } \left(\frac{q}{p}\right)=-1
\end{cases},
\]
where $(\tfrac{q}{p})$ is the Legendre symbol. 
\end{definition}

We note that in the former case, the Cayley graph of $\mathrm{PSL}(2,\mathbb{F}_p)$ with generating $S$ has $\tfrac{p(p^2-1)}{2}$ vertices and is non-bipartite, while in the latter case, the Cayley graph of  $\mathrm{PGL}(2,\mathbb{F}_p)$ with generating set $S$ is bipartite with $p(p^2-1)$ vertices. 

Using advanced number-theoretic techniques, Lubotzky, Phillips, Sarnak showed their construction has largest nontrivial adjacency eigenvalue at most $2\sqrt{q}$ and hence is Ramanujan.  Additionally, they also showed their construction has other extremal combinatorial properties, such as having girth (i.e. the length of the shortest cycle) of $\Omega(\log_q{n})$. From a computational standpoint, the LPS construction allows for explicit querying of vertex-neighborhoods, which is a desirable property for analyzing exponentially large graphs.

The LPS construction may be used to generate infinite families of $(q+1)$-regular Ramanujan graphs; however, only for $q$ prime with $q\equiv 1 \pmod{4}$ and $n$ as function of $p$ as given above. That is, despite having outstanding properties, the LPS construction is limited to Ramanujan graphs only of a certain degree -- and for each such particular degree, only to a certain number of vertices $n$. In 1994, Morgenstern \cite{Morgenstern1994} partially ameliorated this restriction by extending the LPS construction to accommodate any prime power $q$, while showing this extended construction is still Ramanujan and satisfies all other combinatorial properties of the LPS graphs. Nonetheless, this still left open the general case of a given degree $k$ and size $n$.

\subsubsection{Marcus, Spielman, Snivrasta construction}

In 2013 and 2015, Marcus, Spielman, Snivrasta gave new constructions of Ramanujan graphs using a new technique called the method of interlacing polynomials. Unlike the LPS construction, Marcus, Spielman, and Snivrasta's first construction \cite{Marcus2013} is valid for any given degree $k$, and second construction \cite{Marcus2015} is valid both for any $k$ and number of vertices $n$. In both cases, their constructions can only be used to generate {\it bipartite} Ramanujan graphs. 

While their interlacing family method implicitly suggests an algorithm to find an MSS graph, such an algorithm would require computing partially specified expected characteristic polynomials, for which no known polynomial time algorithms are known \cite{Cohen2016}. However, in \cite{Cohen2016}, Cohen provided a polynomial time algorithm for computing such polynomials, thereby giving a deterministic algorithm that, for given degree $k$ and even positive integer $n$, returns a bipartite Ramanujan graph, according to the construction given in \cite{Marcus2015}, in polynomial time.

\subsection{{Related work in high-performance computing}}

{
 Due to the aforementioned relationships between graph expansion and other properties important in network design, many proposed HPC network topologies consider graph expansion {\it implicitly}, making a comprehensive review of related work difficult. Before proceeding, we briefly survey related work that {\it explicitly} considers Ramanujan graphs or related expanders as network topologies in contexts pertinent to supercomputing. Perhaps most notably, in the context of datacenter architecture design, Valadarsky et al. \cite{Valadarsky2016} propose ``Xpander",  which utilizes LPS graphs and the theory of graph lifts \cite{Bilu2006}. They evaluate Xpander theoretically, via simulation, and using a network emulator, finding that Xpander outperformed traditional data-center designs; see \cite{xpanderProjectPage} for more. In the early 1990s, Upfal \cite{Upfal1992} applied the theory of $(\alpha,\beta,n,d)$-expander graphs to construct so-called ``multibutterfly" networks. Later, Brewer, Chong, and Leighton \cite{Brewer1994} proposed a hierarchical expander construction, as a means to mitigate wiring complexity; they analyzed the fault tolerance of their so-called ``metabutterfly" topology through simulation against the aforementioned multibutterfly. In optical network design, Paturi et al. \cite{Paturi1991} proposed using expander graphs to interconnect processors, and subsequently analyze parallel algorithms for sorting, routing, associative memory, and fault-tolerance. Lastly, in the context of sensor networks, Kar and Moura \cite{kar2006ramanujan} propose using Ramanujan LPS graphs as communication networks supporting distributed decision making, and test their performance on the convergence speed of distributed consensus. 
}

%

\section{Spectral Gap in Supercomputing Topologies}\label{sec:spec}
Here, we survey a variety of supercomputing topologies.  In addition
to giving formal, and in some cases new or generalized, descriptions
of the underlying graphs, we focus on analyzing their spectral gap,
bisection bandwidth, and diameter. We first consider  grid-like and
grid variant topologies: the hypercube, generalized grid, torus,
butterfly, cube connected cycles and Data Vortex. Then, we consider
several miscellaneous topologies: the CLEX, DragonFly,
$G$-connected-$H$, and SlimFly topologies. Our results on algebraic
connectivity and bisection bandwidth are summarized in Table
\ref{T:boundsBW}.

Before proceeding, we first establish a key algebraic tool that we utilize frequently, allowing us to compute subsets of a given graphs spectra through that of a simpler, ``reduced'' graph. 

\begin{lemma}[Reduction Lemma]
Let $G$ be a graph and let $\Gamma$ be a subgroup of $\Aut(G)$, {the automorphism group of $G$}.  Let
$H$ be a weighted, directed, looped graph with vertex set given by 
the orbits of $\Gamma$ over $G$ and where the weight of edge from
orbit $\sigma$ to orbit $\tau$ is the total weight of an arbitrary
vertex $v$ in the orbit $\sigma$ to the orbit $\tau$.  The spectrum of $H$ is a subset of
the spectrum of $G$.  Furthermore, any eigenpair $(\lambda,v)$ of
$G$ such that $\lambda$ is not an eigenvalue of $H$ has the property
that $v$ sums to zero along orbits of $\Gamma$.
\end{lemma}

\begin{proof}
Let $(\lambda,w)$ {be a right eigenpair} of $H$.  We define the vector
$w^{\Gamma}$ as follows; for any vertex $v$ in
$G$, define $e_v^Tw^{\Gamma} = e_{\sigma}^Tw$ where $\sigma$ is the
orbit containing $v$.  Now let $S$ be the collection of orbits and
suppose $v$ is in orbit $\tau$.  We then have that 
\begin{align*}
e_v^TA_G w^{\Gamma} &= \sum_j e_v^TA_Ge_je_j^Tw^{\Gamma} \\
&= \sum_{\sigma \in S} \sum_{j \in \sigma} e_v^TA_Ge_je_j^Tw^{\Gamma} \\
&= \sum_{\sigma \in S} \sum_{j \in \sigma} e_v^TA_Ge_je_{\sigma}^Tw\\
&= \sum_{\sigma \in S} e_{\tau}^TA_He_{\sigma}e_{\sigma}^Tw \\
&= \lambda e_{\tau}^Tw \\
&= \lambda e_{v}^Tw^{\Gamma}.
\end{align*}
As $v$ is arbitrary we have that $\lambda$ is also an eigenvalue of
$G$.  

Now suppose that $(\lambda,v)$ is an eigenpair for $G$ and consider
$v^{\Gamma} = \sum_{\sigma \in \Gamma} v_{\sigma}$.  Since each
$\sigma$ is an automorphism of $G$, $v_{\sigma}$ is also an
eigenvector with eigenvalue $\lambda$.  Thus $v^{\Gamma}$ is either an
eigenvector with eigenvalue $\lambda$ or it is the zero vector.  Since
$v^{\Gamma}$ is constant over orbits of $\Gamma$, we can form
$v^{\Gamma}_H$ as the vector of values over orbits.  It is clear that
$A_H v^{\Gamma}_H = \lambda v_H^{\Gamma}$ and so either $\lambda$ is
in the spectrum of $H$ or $v_H^{\Gamma}$ is zero and $v$ sums to zero
over orbits of $\Gamma$. 
\qed \end{proof}

{We illustrate an application of the Reduction Lemma to a fat tree topology in Figure \ref{fat_tree}}. We note that the Reduction Lemma is almost certainly not new.  In fact, it can be viewed as a special case of several other results on describing the interlacing of spectra of a matrix with a quotient matrix, see for instance \cite[Chapter 1]{Brualdi:GraphsandMatrices} and \cite[Chapter 2]{Brouwer:SpectraGraphs}.  

\begin{figure}[b]
  \centering
  \begin{tikzpicture}[vertex/.style={circle,draw,fill=yellow,minimum size=8}]
    \node[vertex] (root) at (3.75,4) {};
    \foreach \x in {0,...,15}
    {
      
      \node [vertex] (0\x) at (.5*\x, 0) {};
    }
    \foreach \x in {0,...,7}
    {
      \node[vertex] (1\x) at (.25 + \x,.75) {};
      \pgfmathtruncatemacro{\left}{2*\x}; 
        \pgfmathtruncatemacro{\right}{2*\x+1}; 
      \draw (0\left) -- (1\x) -- (0\right);
    }
    \foreach \x in {0,...,3}
    {
      \node[vertex] (2\x) at (.75+2*\x,1.75) {};
      \pgfmathtruncatemacro{\left}{2*\x}; 
        \pgfmathtruncatemacro{\right}{2*\x+1}; 
      \draw[line width = 0.8] (1\left) -- (2\x) -- (1\right);

    }
    \foreach \x in {0,...,1}
    {
      \node[vertex] (3\x) at (1.75+4*\x,2.75) {};
      \pgfmathtruncatemacro{\left}{2*\x}; 
        \pgfmathtruncatemacro{\right}{2*\x+1}; 
      \draw[line width = 1.6] (2\left) -- (3\x) -- (2\right);
      \draw[line width = 3.2] (3\x) -- (root);
    }

    \node[vertex] (0) at (11,0) {};
    \node[vertex] (1) at (11,.75) {};
    \node[vertex] (2) at (11,1.75) {};
    \node[vertex] (3) at (11,2.75) {};
    \node[vertex] (4) at (11,4) {};

    \draw[->,line width = 0.8] (0) to[out=45,in=-45] node[right]{1}
    (1);
    \draw[->,line width = 0.8] (1) to[out=45,in=-45] node[right]{2}
    (2);
    \draw[->,line width = 0.8] (2) to[out=45,in=-45] node[right]{4}
    (3);
    \draw[->,line width = 0.8] (3) to[out=45,in=-45] node[right]{8}
    (4);

    \draw[->,line width = 0.8] (4) to[out=-135,in=135] node[left]{16}
    (3);
    \draw[->,line width = 0.8] (3) to[out=-135,in=135] node[left]{8}
    (2);
    \draw[->,line width = 0.8] (2) to[out=-135,in=135] node[left]{4}
    (1);
    \draw[->,line width = 0.8] (1) to[out=-135,in=135] node[left]{2}
    (0);

    \node (left) at (7.75,1.75) {};
    \node (right) at (10.25,1.75) {};
    \draw[->,line width = 1.2] (left) to node[above]{reduction}
    node[below]{lemma}  (right);
  \end{tikzpicture}
  \caption{{An application of the Reduction Lemma to the fat tree.}}\label{fat_tree}
\end{figure}
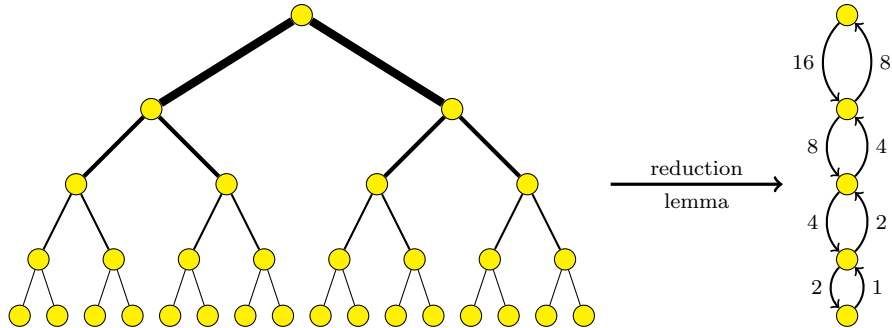

%

It is worth noting that several of the topologies we will consider
have implementations which have minor irregularities in the node
radixes.  These 
deviations from regularity have little effect on the true performance
of the network and so we will add self-loops as needed to eliminate 
irregularity and simplify the analysis. This will not change the nature of any of our results as the bisection bandwidth and diameter both are unaffected by arbitrary self-loops.

\subsection{Product (Grid-Like) Topologies}
For a number of high-dimensional supercomputing topologies, their underlying graphs can be obtained via repeated {\it graph products}. Product graphs are highly structured and possess properties which can sometimes be tightly controlled by those of their factor graphs. Below, we briefly describe three such topologies: the hypercube, torus, and generalized grid. These graphs are obtained via a particular graph product called the {\it Cartesian product}, denoted $G \Osq H$. The graph $G \Osq H$ is on vertex set $V(G) \times V(H)$ and is defined by the edge condition: $(u,u')$ and $(v,v')$ are adjacent if and only if either
\begin{itemize}
\item $u=v$ and $\{u',v'\} \in H$, or
\item $u'=v'$ and $\{u,v\} \in G$. 
\end{itemize}

We note that the adjacency matrix of $G \Osq H$, can be written succinctly in terms of those of $G$ and $H$,
\[
A_{G \Osq H}=A_G \otimes I + I \otimes A_H,
\]
where $I$ denotes the identity matrix and $\otimes$ denotes Kronecker product. Using the above characterization, it is easy to show that the adjacency (or Laplacian) eigenvalues of $A_{G \Osq H}$ consists of $\lambda(G)_i + \lambda(H)_j$ over all $1 \leq i \leq |V(G)|$ and $1 \leq j \leq |V(H)|$; hence the $d$-fold Cartesian product eigenvalues consists of all possible $d$-sums of the factor graph eigenvalues. In particular, the algebraic connectivity of $G \Osq H$ is the minimum of the algebraic connectivity of $G$ and the algebraic connectivity of $H$.


\begin{definition}[Hypercube, $Q_d$]
The $d$-dimensional hypercube, $Q_d$, is on $n=2^d$ vertices, defined by the $d$-fold Cartesian product $P_2 \Osq \dots \Osq P_2$, where $P_2$ is the path with 1 edge. 
\end{definition}

It is well-known that $Q_d$ has algebraic connectivity of $2$ and bisection bandwidth $2^{d-1}=n/2$. The hypercube is a special case of the generalized grid graph, defined below.

\begin{definition}[Generalized Grid, $G_{k_1,\dots,k_d}$]
The $d$-dimensional, generalized grid the $d$-fold Cartesian product $P_{k_1} \Osq \cdots \Osq P_{k_d}$, where $P_{k_i}$ is the path of length ${k_i-1}$. 
\end{definition}

We note that taking $d=2$, $k_1=m$, and $k_2=n$ yields what is sometimes simply referred to as a grid graph, or $m\times n$ lattice, while taking $k_1=\dots=k_d=2$ yields $Q_d$. Using the aforementioned fact relating the Cartesian product eigenvalues to those of the factor graphs, it is easy to see the algebraic connectivity of $G_{k_1,\dots,k_d}$ is $2-2\cos(\pi/\max\{k_1,\dots,k_d\})$. Finally, we define the discrete torus topology, which is given by the cartesian product of cycles. 

\begin{definition}[Torus, $C^d_k$]
The discrete torus $C^d_k$ is the $d$-fold graph box product of a $k$-cycle, i.e. $C_k \Osq \dots \Osq C_k$. This graph is regular on $n=k^d$ vertices, and has degree $2d$. 
\end{definition}

It is not difficult to show the algebraic connectivity of the torus $C^d_k$ is $2(1-\cos(\nicefrac{2\pi}{k}))$.

\subsection{Grid Variants} 
The collection of topologies we consider in this section are closely
related to topologies formed from the product operation, but with
minor twists or modifications.  Oftentimes these toplogies start from
some grid-like layout and permute the connections or add small
substructures to achieve desired properties.

\subsubsection{Butterfly}

One of the more well known grid variants is the Butterfly
topology~\cite{Leighton:Intro}.  In its most simple form the Butterfly
topology consists of a sereis of shuffling layers based on the binary
representation of the node names.  More concretely, there are 
  $p\log_2(p)$ switches arranged in a
  $\log_2(p)$-by-$p$ array of $p$ switches in one of $\log_2(p)$
  ranks. For each rank, each of the $p$ switches is
  connected to two switches in the previous rank and two switches in
  the next rank. 
  The nodes in rank $i$ and position $j$ are connected to the switch
  $j$ and switch $m$ in rank $i-1$, where $m$ is formed by flipping
  the $i^{\textrm{th}}$ bit in binary representation of $j$.  It is
  also connected to switch $j$ and $m'$ in rank $i+1$, where $m'$ is
  formed by switching the $(i+1)^{\textrm{st}}$ bit in $j$.  The
  Butterfly topology has diameter $\log_2(p)$ and bisection width
  $\nicefrac{p}{2}$

\begin{definition}[Butterfly, $\BF(k,s)$] The $k$-ary, $s$-fly
  butterfly network where there are $s$-layers of switches, and each
  switch has $k$ ``forward'' connections.  More concretely, the
  switches can be indexed by elements of $[s] \times [k]^s$.  The
  ``forward'' connections from $(i, (a_1,\ldots,a_{s}))$ to $(i+1,
  (a_1',\ldots,a_s'))$ are formed by keeping all but the
  $i^{\textrm{th}}$ component of $a$ fixed, i.e.\ $a_j = a_j'$ if $j
  \neq i$.  Depending on the application, the $s$ layers can either be
  connected linearly (no connection from layer $s$ to layer $1$), or
  cyclically (connection from layer $s$ to layer $1$).  For convenience,
  we will restrict ourselves to the cyclic arrangement.
\end{definition}  

It is straightforward to see that these networks have a diameter of
$s$ by considering two elements in the same layer, $(i,\mathbf{a})$ and
$(i,\mathbf{b})$, where no coordinate of $\mathbf{a}$ and $\mathbf{b}$
agree.

\begin{proposition}
Let $G$ be a $k$-ary, $s$-fly Butterfly network.  The bisection
bandwidth of $G$ is at most $\frac{(k+1)k^s}{2}$ and the algebraic
connectivity is at most $2k - 2k\cos\paren{\frac{2\pi}{s}}$.
\end{proposition}

\begin{proof}
To upper bound the bisection bandwidth we first consider the case
where $k$ is even and define $X = \left[\frac{k}{2}\right] \times
[k]^{s-1}$.  The bipartition  we consider is then $([s] \times X, [s]
\times \overline{X})$.   In order for $(s,x) \in  [s]\times X$ and
$(s,x') \in [s]\times
\overline{X}$ to be adjacent, it must be the case that $\set{s,s'} =
\set{1,2}$ and $x$ and $x'$ differ only in the first coordinate. This
gives that $e([s]\times X, [s]\times \overline{X}) = 2\paren{\frac{k}{2}}^2k^{s-1} =
\frac{k^{s+1}}{2}$.

When $k$ is odd the we construct a slightly more complicated
partition.  In particular, for $0 \leq i \leq s-1$ define $X_i =
\set{\frac{k+1}{2}}^i \times \left[\frac{k-1}{2}\right] \times [k]^{s-1-i}$ and
let $X = \bigcup_i X_i$.  Note that
\[ \size{X} = \sum_{i=0}^{s-1} \size{X_i} = \frac{k-1}{2}
  \sum_{i=0}^{s-1} k^{s-1-i} = \frac{k-1}{2} \frac{k^s - 1}{k-1} =
  \frac{k^s-1}{2}.\]  In particular $([s] \times X, [s] \times
\overline{X})$ is a bipartition of the vertex set of the $k$-ary,
$s$-fly Butterfly network.  Now to evaluate the bisection bandwidth we
wish to count pairs 
$(u,v) \in X_i \times \overline{X}$ such that $u$ and $v$ differ in precisely one
component.  If we fix some $x_i \in X_i$, then for all $1 \leq j \leq
i+1$, modifying the $j^{\textrm{th}}$ component to be in
$\frac{k+3}{2}, \ldots, k$ yields such a pair.  We note that modifying
any entry $j > i+1$ will preserve membership in $X_i$.  Thus the only
remaining case to consider is when index $(i+1)$ is modified to have value
$\frac{k+1}{2}$.  This takes us outside the set $X$ if and only if $x_i \not\in
\set{\frac{k+1}{2}}^i \times \left[\frac{k-1}{2}\right] \times
  \set{\frac{k+1}{2}}^j \times \left[\frac{k-1}{2}\right] \times
  [k]^{s-2-i-j}$ for some $0 \leq j \leq s- 2- i$. As there are
  $\frac{k-1}{2} \frac{k^{s-i-1}+1}{2}$ such terms $x_i$, this gives
  that the total number of pairs $
  (u,v) \in X_i \times \overline{X}$ which differ by exactly one component is given by
  \[   \paren{\frac{k-1}{2}}k^{s-i-1}(i+1)\paren{\frac{k-1}{2}} +
    \paren{\frac{k-1}{2}}\paren{\frac{k^{s-i-1}+1}{2}} =
   \frac{\paren{k-1}\paren{ (i+1)k^{s - i} - ik^{s-i-1}+1}}{4}.\]
  Thus the bisection bandwidth is at most
  \begin{align*}
\sum_{i=0}^{s-1} 2 \frac{\paren{k-1}\paren{ (i+1)k^{s - i} -
    ik^{s-i-1}+1}}{4} &= \frac{s(k-1)}{2}+ \frac{k-1}{2}
                        \sum_{i=0}^{s-1} (i+1)k^{s-i} - ik^{s-i-1} \\
    &= \frac{s(k-1)}{2} + \frac{k-1}{2} \paren{k^s +
      2\frac{k^s-k}{k-1} - (s-1)} \\
                      &= \frac{k-1}{2} +
                        \frac{k^{s+1}-k^s}{2} + k^s - k \\
    &= \frac{k^{s+1}+k^s -k - 1}{2}.
  \end{align*}

  To upper bound the algebraic connectivity, we note that there is an
  automorphism group of the Butterfly topology in which the orbits are
  given by the layers.  Thus, by applying the reduction lemma to this
  automorphism group we get an $s$-cycle with edge multiplicity $k$.
  The bound on the algebraic connectivity then follows immediately.
  \qed \end{proof}

\subsubsection{Data Vortex}
The Data Vortex topology was designed as a ``streaming" topology with the idea that all of the data is constantly in motion (i.e., it is never buffered) and the data swirls from processors in the outer ring of the vortex 
towards the processors on the inside of the topology~\cite{Hawkins2007,Shacham2005}.  This streaming methodology has allowed the Data Vortex topology to handle the transmission of high-volumes of data without suffering from signficant congestion related performance degradation (see \cite{Gioiosa2017,Gioiosa2016,Iliadis2007,Yang2002} for a more in-depth discussion of the performance benefits of the Data Vortex topology.). Formally, the topology is defined a series concentric cylinders with ``angular" transitions between them.  Within the cylinders there is a switching topology reminiscent of the layers of the 2-ary Butterfly topology.  More concretely, we have:

\begin{definition}[Data Vortex, $\DV(A,C)$]
The Data Vortex topology with parameters $A,C$ is a graph with vertex set $\Z_{A} \times \Z_{C} \times \Z_2^{C-1}$, and edge set given by:
\begin{enumerate}
\setlength{\itemsep}{.5em}
\item for all $(a,c,h) \in \Z_A \times \Z_C \times \Z_2^{C-1}$ there is an edge to $(a+1,c+1,h)$,
\item for all $(a,c,h) \in \Z_A \times \paren{\Z_C-\set{0}} \times \Z_2^{C-1}$ there is an edge to $(a+1,c,h + e_c)$ where $e_c$ denotes the unit vector for the $c^{\textrm{th}}$ component of $\Z_2^{C-1}$, and 
\item for all $(a,c,h) \in \Z_A \times \set{0} \times \Z_2^{C-1}$ there is an edge to $(a+1,c,h) = (a+1,0,h)$.
\end{enumerate}
\end{definition}
Although the Data Vortex is designed as a streaming topology (and is in particular, indirect), we will consider it as a direct topology in which each node denotes a compute node.  
    
%

\begin{proposition}
The algebraic connectivity of the Data Vortex topology with $A$ angles and height $H$ is at most $\min\set{2 - 2\cos\paren{\frac{\pi}{C}}, 2-2\cos\paren{\frac{2\pi}{A}}} = \bigOh{ \frac{1}{ \max\set{A^2,(C-1)^2}}}$.  Furthermore, the bisection bandwidth is at most $A2^{C-2}$.  
\end{proposition}

\begin{proof}
We begin by first noting that the vertices in the outer and inner ring of the Data Vortex have degree 3, so we will consider the topology formed by adding a self loop to each of these vertices.  Alternatively, we could add an edge between corresponding vertices in the inner and outer ring by observing that in typical use cases these vertices are connected to a common system, forming the ``input'' and ``output'' ports of the system.  However, this modification results in essentially the same asymptotic behavior, so we choose the self-loop modification as it is requires no assumptions about how the Data Vortex interacts with the processing layer of the overall system.  

We first consider the bisection bandwidth by separating the vertices based on height, specifically partitioning into vertices of height $1, \ldots, 2^{C-2}$, and those of height $2^{C-2}+1, \ldots, 2^{C-1}$.  Clearly this is a bisection.  As no edge between concentric rings changes height, it suffices to consider only those edges internal to a ring.  However, as only one ring flips the leading bit of the height vector, this gives that the bisection bandwidth is at most $A2^{C-2}$.  

In order to bound the algebraic connectivity, we will apply the reduction lemma.  Specifically we consider the automorphism group generated by the bit-flip operations on the height.  As these act uniformly on the height the edges between successive rings are clearly preserved.  Further, as the bit-wise differences are preserved by the bit-flip operations, this preserves edges on each ring.  Under this automorphism group, the Data Vortex topology reduces to $C_A \Osq P'_{C}$ where $P'_k$ is the $k$-vertex path with loops at each end.  The result bounding the algebraic connectivity follows immediately.    
\qed \end{proof}

\subsubsection{Cube Connected Cycles}\label{SS:CCC}
Loosely speaking, the Cube Connected Cycles (CCC) graph consists of a hypercube in which each vertex has been replaced by a cycle. Preparata and Vuillemin \cite{Preparata1981} proposed the Cube-Connected Cycles as a versatile network topology for connecting processors in a parallel computer, which emulates the the robust connectivity properties of the hypercube, but (due to the cycle modification) only requires three connections per processor. They conclude that ``by combining the principles of parallelism and pipelining, the CCC can emulate the cube connected machine and shuffle-exchange network with no significant degradation in performance."

As suggested in \cite{Riess2012}, CCC graph is a special case of a more general graph construction in which an arbitrary graph is connected in a hypercube structure. More precisely:

\begin{definition}[Cube Connected Cycles, $\CCC(d)$]
The {\it $d$-dimensional cube-connected graph} of a given graph $G$, denoted $\CC(G,d)$, has vertex set $V(G) \times \{0,1\}^d$ and edge condition $(v_i,x) \sim (v_j,y)$ if and only if 
\begin{itemize}
\item $v_i \sim v_j$ in $G$, or
\item $v_i=v_j$ and the hamming distance between $x$ and $y$ is 1. 
\end{itemize}

\end{definition}

Taking $G=C_d$ yields the well-known {\it Cube-Connected Cycles} graph. Riess, Strehl, and Wanka proved the following result, which relates the characteristic polynomial of $\CC(G)$ to those of loop-weighted variants of $G$:

\begin{theorem}[Riess, Strehl, Wanka \cite{Riess2012}] \label{thm:blkCirculantPoly}
Let $G$ be an $d$-vertex graph. For $\vec{s}=(s_1,\dots,s_d) \in \{-1,1\}^d$, let $G[\vec{s}]$ denote the graph obtained from $G$ by adding a loop of weight $s_i$ to each vertex $i$. Then
\[
\chi(\CC(G,d))=\prod_{\vec{s} \in \{-1,1\}^d} \chi(G[\vec{s}]),
\] 
where $\chi(G)$ denotes the characteristic polynomial of the adjacency matrix of $G$. 
\end{theorem}

As an immediate consequence, we have that the spectral set of $\CC(G,d)$ is the union of the spectral sets of $G[\vec{s}]$ over all $\vec{s} \in \{-1,1\}^d$. Using their result, we can derive good estimates of the spectral expansion of the CCC. To do so, we first prove the following lemma. 

\begin{lemma} \label{claim:oneNeg}
Let $G$ be a connected, $n$-vertex graph. The second largest adjacency eigenvalue of $\CC(G,d)$ is the maximum eigenvalue of $G[\vec{s}^*]$, where $\vec{s}^*=(s_1,\dots,s_d) \in \{-1,1\}^d$ is such that for some fixed $j \in [n]$, $s_j=-1$ and for all other $i\not= j$, $s_i=1$.
\end{lemma}

In the proof of Lemma \ref{claim:oneNeg}, we will use the following basic fact:
\begin{fact} \label{fct:plusesGood}
Let $G$ be a connected, $n$-vertex graph. Let $\vec{r},\vec{t} \in \{-1,1\}^d$, $\vec{r} \not= \vec{t}$, be such that $\vec{r}$ agrees with $\vec{t}$ on any $i \in [n]$ where $t_i=1$, and let $i_1,\dots,i_k \in [n]$ denote indices on which they differ, i.e. where $t_{i_j}=-1$ and $r_{i_j}=1$ for $j \in [k]$. Then the largest adjacency eigenvalue of $G[\vec{r}]$ is strictly greater than that of $G[\vec{t}]$.
\end{fact}

\begin{proof}
Let $A$ and $A'$ denote the adjacency matrices of $G[\vec{r}]$ and $G[\vec{t}]$, respectively, and let
 $\vec{x}=(x_1,\dots,x_n)$ denote the normalized, dominant eigenvector of $A'+I$, whose entries are all positive by the Perron-Frobenius theorem. By definition, we have $\vec{x}^T (A+I) \vec{x} \ - \ \vec{x}^T (A'+I) \vec{x} = 2\sum_{j=1}^k x_{i_j}^2>0.$
\qed \end{proof}

\begin{proof}[Proof of Claim \ref{claim:oneNeg}]

From Theorem \ref{thm:blkCirculantPoly} and Fact \ref{fct:plusesGood}, we have that the largest adjacency eigenvalue of $\CC(G,d)$ is that of $G[\vec{1}_d]$, and furthermore that if $\vec{1}_d \not=\vec{t} \in \{-1,1\}^d$ does not satisfy the property in the claim, then there exists some $\vec{s}$ that does, which we denote $\vec{s}^*$, such that $\lambda_{1}(G[\vec{t}])< \lambda_{1}(G[\vec{s}^*])$.
So, let $G$ and $G'$ denote $G[\vec{1}_d]$ and $G[\vec{s}^*]$, respectively, on vertex set $\{v_1,\dots,v_n\}$, where $v_j$ denotes the vertex in $G[\vec{s^*}]$ with a loop of weight $-1$. Labeling the adjacency eigenvalues $\lambda_1 \geq \dots \geq \lambda_{n} $, all that remains to show is that
\begin{align}
\lambda_{2}(G) < \lambda_{1}(G'). \label{eqn:strictInq}
\end{align}
By Cauchy's interlacing theorem, if we delete $v_j$ from $G$ and $G'$, we have
\begin{align*}
\lambda_{2}(G) &\leq \lambda_{1}(G\setminus v_j)  \leq \lambda_{1}(G), \\
\lambda_{2}(G') &\leq \lambda_{1}(G'\setminus v_j)  \leq \lambda_{1}(G').
\end{align*}
But since $G\setminus v_j = G'\setminus v_j$, combining the above inequalities yields
\[
\lambda_{2}(G) \leq \lambda_{\max}(G\setminus v_j)  \leq \lambda_{1}(G').
\]
To see the inequality in $(\ref{eqn:strictInq})$ is strict, assume for contradiction that $\lambda_{1}(G'\setminus v_j)= \lambda_{1}(G')$. Then if $\vec{x}=(x_1,\dots,x_j,\dots,x_n)$ denotes the dominant eigenvector of $A'$ associated with $\lambda_{1}$, this implies if we set $x_j=0$, the vector $(x_1,\dots,0,\dots,x_n)$ is still an eigenvector of $G'$ associated with $\lambda_{1}$. But applying the Perron-Frobenius theorem to $A'+I$ yields that the dominant eigenvector of $A'+I$ (and hence that of $A'$) is unique and has all entries positive, which is a contradiction. 
\qed \end{proof}

Using Lemma \ref{claim:oneNeg}, we have:

\begin{proposition} 
The algebraic connectivity of the $d$-dimensional, cube-connected cycles is at most on the order of $2\left(1-\cos\left(\frac{\pi}{d+2}\right) \right)$.
\end{proposition}

\begin{proof}
By Lemma \ref{claim:oneNeg}, it suffices to consider the largest adjacency eigenvalue $\lambda_{1}$ of the $d$-cycle with one loop of weight $-1$ on one vertex, and loops of weight 1 on other vertices. Letting $A'$ denote this graphs adjacency matrix, a routine calculation shows that for ${\bf x}=(x_1,\dots,x_n)$ defined by $x_i=\sin\left(\frac{\pi i}{d+2} \right)$,
\begin{align*}
\lambda_{1}(A')\geq \frac{\langle \vec{x}, A' \vec{x} \rangle}{\langle \vec{x},\vec{x} \rangle}&=2\cos\left(\frac{\pi}{d+2}\right)+1+\frac{\sin^2(\tfrac{\pi}{d+2})\left(2 \cos(\tfrac{\pi}{d+2})-2 \right)}{\tfrac{d+1}{2}+\cos(\tfrac{2\pi}{d+2} )}. 
\end{align*}
We note that above expression is strictly larger than the second largest adjacency eigenvalue of the d-cycle with all loops of weight 1, for $d\geq 2$.

\qed \end{proof}

It is worth mentioning that the Cube Connected Graphs are really a
specific instance of a more general technique of constructing
supercomputing topologies, which we refer to as $G$-connected-$H$.
As the generic $G$-connected-$H$ topologies are not grid-like
topologies we will defer their discussion to Section \ref{SS:GconnH}.

\subsection{Miscellaneous}
In this section we consider a few topologies that do not (necessarily) have a strong grid structure.  Typically these topologies have some sort of recursive or multi-layer structure in order to attempt to combine ``good'' properties of several types of graphs.

\subsubsection{CLEX}
``Clique-Expander" (CLEX) is a new supercomputing topology recently introduced by Lenzen and Wattenhofer \cite{Lenzen2016}. The CLEX construction is recursive, starting with a specified number of cliques that are sparsely interconnected. According to the authors, the CLEX design is motivated by a desire to ``localize the issue of an efficient communication network to much smaller systems which may reside on a single multi-core board". CLEX is touted to have superior point-to-point communication properties, particularly when compared with toroidal topologies; nonetheless, the Lenzen and Wattenhofer acknowledge ``the price we pay for these properties are [high] node degrees". 

In this section, we will define the CLEX topology, and prove new bounds on the diameter, algebraic connectivity, and bisection bandwidth. As the authors of CLEX note that that ``the high connectivity of a CLEX system could be considered an abstraction that can be replaced by any efficient local communication scheme within the cliques", we generalize our spectral analysis of CLEX accordingly. In particular, our analysis allows one to replace the cliques of the CLEX construction with other graphs.  We first begin by defining the CLEX graph, as given in \cite{Lenzen2016}.

\begin{definition}[CLEX, $C(k,\ell)$]
For given positive integers $k$ and $\ell$, a CLEX digraph, denoted $C(k,\ell)$, is on $n=k^{\ell}$ vertices with $\ell$ ``levels", and is defined recursively. The base case is $C(k,1)=K_k$, the complete graph on $k$ vertices. The vertex set of $C(k,\ell+1)$ is the $(\ell+1)$-fold cartesian product of $V(K_k)$. The edge set of $C(k,\ell+1)$ consists of all edges from $k$ copies of $C(k,\ell)$, with additional directed edges between these copies. Note the last entry in each vertex identifies which ``copy" of $C(k,\ell)$ that vertex belongs to. The additional edges between these copies of $C(k,\ell)$ are given by the set:
\[
\{((v_1,\ldots,v_l, i), \hspace{1mm} (v_1,\ldots,v_{\ell-1},j,v_{\ell})) : i,j \in [k]\}.
\]
\end{definition}

With regard to the diameter of CLEX, the authors in \cite{Lenzen2016} give an upper bound\footnote{note that there is actually a typo in their paper here, as they write $2^{1/\ell}-1$, which is non-integer} of $C(k,\ell)$ as $2^{\ell}-1$. We claim that the diameter is bounded by $\ell$. 

\begin{proposition}
The diameter of the CLEX graph $C(k,\ell)$ is at most $\ell$. Furthermore, this bound is tight. 
\end{proposition}
\begin{proof}
We construct a walk of length $\ell$ between two arbitrary vertices of $C(k,l)$, $(v_1,\ldots,v_l)$ and $(w_1,\ldots,w_{\ell})$ as follows:
\[
(v_1,\ldots,v_{\ell}),
(w_{\ell},v_2,\ldots,v_{\ell}),(w_1,w_l,v_3,\ldots,v_{\ell}),(w_1,w_2,w_l,v_4,\dots,v_{\ell}),\dots,(w_1,\ldots,w_{\ell-1},w_{\ell}).
\]

Furthermore, this bound can seen to be tight by considering the path
between $(i,i,\ldots,i)$ and $(j,j,\ldots, j)$ for any $i \neq j$.
Specifically, although each edge can modify up to two positions in the
vector describing the vertex, it can change the count of any
particular symbol in the string by at most one. 
\qed \end{proof}

The rest of our analysis will consider the CLEX digraph as an
undirected multi-graph (potentially with loops).  Specifically, for
every directed edge $(i,j)$ in the CLEX digraph we will have an
undirected edge $\set{i,j}$ and thus the total degree of any vertex
does not change.  As our analysis only relies minimally on the
structure of $K_k$, we will consider a generalized version of CLEX,
denoted $C(G,\ell)$ where $G$ is a $t$-regular, connected graph on $k$
vertices.  We note that both the regularity and connectivity
conditions can be relaxed at various points in the following analysis,
however we make both assumptions for simplicity of presentation.  

We first note that even when $G \neq K_k$, the arguments regarding the
diameter follow exactly after accounting for the diameter of $G$ and
potentially directed nature of $G$.

\begin{lemma}
Let $G$ be a $k$-vertex graph, then 
\[ C(G,\ell) = G \otimes I_{k^{\ell-1}}  + \sum_{j=0}^{\ell-2} I_{k^j}
  \otimes M \otimes I_{k^{\ell-2-j}},\] where $M \in \Z^{k^2\times k^2}$
is given by 
\[ M_{(i,j),(a,b)} = \begin{cases} 2 & i = b, j = a \\ 1 & i = b, j
    \neq a \\ 1 & i \neq b, j = a \\ 0 & \textrm{otherwise}.\end{cases}\]
\end{lemma}
\begin{proof}
The generic formula will follow immediately from the inductive
characterization of the CLEX graphs.  We note that the edges of
$C(G,\ell+1)$ can be partitioned in two sets, those that come from
$C(G,\ell)$ and the cross edges ``between'' copies of $C(G,\ell)$.
Letting $C^G_{\ell}$ be the adjacency matrix for $C(G,\ell)$, the edge
coming from the copies of $C(G,\ell)$ can be described by $C_{\ell}^G
\otimes I_k$.  Now note that an edge is added between
$(v_1,\ldots,v_{\ell-1},v_{\ell},v_{\ell+1})$ and
$(w_1,\ldots,w_{\ell-1},w_{\ell},w_{\ell+1})$ precisely when $v_i =
w_i$ for $1 \leq i \leq \ell_i$ and $v_{\ell+1} = w_{\ell}$ or
$w_{\ell+1} = v_{\ell}$.  Thus the cross edges are given by
$I_{k^{\ell-1}} \otimes M$ and we have that \[ C^G_{\ell+1} =
    C^G_{\ell}\otimes I_k + I_{k^{\ell-1}}\otimes M.\]  The
  non-inductive formula follows immediately from this relationship.
\qed \end{proof}

\begin{lemma}
Let $M \in \Z^{k^2\times k^2}$ be defined by \[ M_{(i,j),(a,b)} = \begin{cases} 2 & i = b, j = a \\ 1 & i = b, j
    \neq a \\ 1 & i \neq b, j = a \\ 0 &
    \textrm{otherwise}.\end{cases}\]
We then have that $\spec{M}$ is the multiset $\set{2k, k^{(k-1)},
  (-k)^{(k-1)}, 0^{(k-1)^2}}$. 
\end{lemma}

\begin{proof}
Let $\set{e_i}$ be the standard basis vectors for $\R^k$ and let
$\one$ be the all ones vector in $\R^k$.  We first note that 
\[ M = \sum_{i = 1}^k \paren{\one \otimes e_i}\paren{e_i\otimes
    \one}^T + \paren{e_i \otimes \one}\paren{\one \otimes e_i}^T. \]
It is easy to see at this point that $\one \otimes \one$ is an
eigenvector of $M$ with eigenvalue $2k$.  Furthermore, we can see that
the non-trivial eigenvectors must lie in $\linspan{ \set{ \one \otimes
    e_i}_i \cup \set{ e_i \otimes \one}_i}$, as a $2k-1$ dimensional
subspace of $\R^{k^2}$.   

Now consider 
\begin{align*}
M \paren{ e_j \otimes \one} &= \sum_{i=1}^k \paren{\paren{\one \otimes e_i}\paren{e_i\otimes 
    \one}^T + \paren{e_i \otimes \one}\paren{\one \otimes 
                              e_i}^T} \paren{e_j \otimes \one} \\
&= \sum_{i=1}^k \paren{\one \otimes e_i}\paren{e_i\otimes 
    \one}^T \paren{e_j \otimes \one} + \paren{e_i \otimes \one}\paren{\one \otimes 
                              e_i}^T \paren{e_j \otimes \one}\\
&= k \paren{\one \otimes e_j} + \sum_{i=1}^k e_i \otimes \one \\
&= k \paren{\one \otimes e_j} + \one\otimes \one.
\end{align*}
Similarly, we have that $M \paren{\one \otimes e_j} = k \paren{e_j
  \otimes \one} + \one \otimes \one$.  From this it easy to see that 
\[ M \paren{  e_j \otimes \one - \one \otimes e_j }= -k \paren{  e_j
    \otimes \one - \one \otimes e_j }, \] for all $j$.  Noting that
$\sum_j  \paren{  e_j
    \otimes \one - \one \otimes e_j } = 0$, we have that this yields a
  $k-1$-dimensional eigenspace associated with the eigenvalue $-k$.
  Finally, we note that $M \paren{ e_j \otimes \one + \one \otimes e_j
    - \frac{2}{k} \one \otimes \one} = k \paren{ e_j \otimes \one +
    \one \times e_j - \frac{2}{k}\one \otimes \one}$, we similarly
  observe a $k-1$ dimensional eigenspace associated with the
  eigenvalue $k$.  As the dimension of the non-trivial eigenspaces is
  at most $2k-1$, this provides a complete characterization of the
  spectrum.  
\qed \end{proof}

\begin{proposition} \label{L:clex_spec}
Let $G$ be a $t$-regular, connected graph on $k$ vertices.  The
algebraic connectivity of $C(G,\ell)$ is at most $t + 3k - 1$. 
\end{proposition}

\begin{proof}
First we note that since $G$ is $t$-regular, $C(G,\ell)$ is $t +
2k\paren{\ell-1}$ regular.  Now let $(\lambda,v)$ be the eigenpair associated with the second largest
eigenspace of $G$ such that $\norm{v}=1$ and let $w = v
\otimes \paren{\frac{1}{\sqrt{k}}\one}^{\otimes \ell - 1}$.   Since $G$ is $t$-regular, we have that $\inner{v}{\one} =
0$ and thus $\inner{ w}{\one^{\otimes \ell}} = 0$.  Furthermore, since
$\norm{v} = 1$ and $\norm{\frac{1}{\sqrt{k}}\one} = 1$, we have that
$\norm{w} = 1$.  Thus $w^TMw$ is a lower bound on the second largest
eigenvalue of $C(G,\ell)$.   We now note that
\begin{align*}
w^T M w &= w^T G \otimes I_{k^{\ell-1}} w + \sum_{i=0}^{\ell-2} w^TI_i
          \otimes M \otimes I_{\ell-2-i} w \\
&= \lambda + (v\otimes \one)^T M (v \otimes M) + \sum_{i=1}^{\ell-2}
  2k \\
&\geq \lambda - k + 2k(\ell-2) \\
 &= -1 -k +2k(\ell-2).
\end{align*}
Thus the spectral gap is at most $t + 2k(\ell-1) - \paren{-1-k +
  2k(\ell-2)} = t + 3k + 1$.
\qed \end{proof}

\begin{proposition} \label{L:clex_bw}
Let $G$ be a $t$-regular connected graph.  If $\ell \geq 3$, the
bisection bandwidth of $C(G,\ell)$ is at most $k^{\ell+1}$.
\end{proposition}
\begin{proof}
We may assume without loss of generality that the vertices of $G$ are
given by $[k] = \set{1,\ldots,k}$ and thus the vertex set of
$C(G,\ell)$ is given by $[k]^{\ell}$.  In order to upper bound the
bisection bandwidth we will provide two explicit partitions of the
vertex set, one for the case when $k$ is even and a modification
construction for when $k$ is odd.  To that end, define $A$ to be the
set of odd integers in $[k]$ if $k$ is even, and in $[k-1]$ if $k$ is
even.  Similarly define $A'$ to be the set of even integers in $[k]$.
We note that if $k$ is even then $[k]$ is a disjoint union of $A$ and
$A'$, while if $k$ is odd $[k]$ is a disjoint union of $A$, $A'$, and
$\set{k}$.

We first consider the case where $k$ is even and define the sets $X =
[k]^{\ell - 2} \times \paren{ A \times A' \cup A' \times A}$ and
$\overline{X} = [k]^{\ell-2} \times \paren{A \times A \cup A' \times A'}$.
Since $\size{A} = \size{A'}$ and $[k]$ is a disjoint union of $A $ and
$A'$, it is clear that $(X,\overline{X})$ is a bisection of the
$C(G,\ell)$.

  Now let $A_{\ell}$ be the adjacency matrix of $C(G,\ell)$ and let
  $\one_X$ (respectively $\one_{\overline{X}}$) be the indicator
  vector for the set $X$ (respectively $\overline{X}$).  By definition
  \[ \BW\paren{C(G,\ell)} = \one_X^T A_{\ell} \one_{\overline{X}} =
    \one^T_X \paren{G \otimes I_{k^{\ell-1}}  + \sum_{j=0}^{\ell-2} I_{k^j}
  \otimes M \otimes I_{k^{\ell-2-j}}} \one_{\overline{X}}.\]
Noting that for any set $S$,  we have that$\one_S^T I_{\size{S}}
\one_{\overline{S}} = 0$ as $S$ and $\overline{S}$ are disjoint, this
can be simplified to
\begin{align*}
  \BW\paren{C(G,\ell)}&= \one_{X}^T\paren{ I_{k^{\ell-3}} \otimes M \otimes 
    I_{k} +  I_{k^{\ell-2}} \otimes M}
                          \one_{\overline{X}} \\
 &= k^{\ell-3} \paren{\one_{[k] \times A \times A'} + \one_{[k] \times A' \times 
   A}}^T 
   \paren{ M \otimes 
    I_{k} +  I_{k} \otimes M} \paren{\one_{[k] \times A 
   \times A} + \one_{[k] \times A' \times A'}} \\
  &= k^{\ell-3} \paren{2\frac{k}{2}\one_{[k] \times A}^TM\one_{[k]
    \times A'} +4 k \one_{ A \times A'}M \one_{A \times A}},
\end{align*}
where the last line comes from the symmetry of $A$ and $A'$ and the
symmetry of $M$ in terms of the Kronecker product.  Substituting in
the defintion for $M$ we get
\begin{align*}
\paren{\one \otimes \one_A}^T M \paren{\one \otimes
  \one_{A'}}&= \sum_{i=1}^{k} \paren{\one \otimes
              \one_A}^T\paren{ (e_i \otimes \one)(\one \otimes e_i)^T
              + (\one \otimes e_i)(e_i \otimes \one)^T}\paren{\one \otimes
              \one_{A'}} \\
  &= \sum_{i=1}^k \size{A} k \paren{e_i^T\one_{A'}} + k \paren{e_i^T\one_A}\size{A'}
  \\
            &= 2 k \size{A}\size{A'} \\
              &= \frac{k^3}{2}
\end{align*}
and
\begin{align*}
\paren{\one_A \otimes \one_{A'}}^T M \paren{\one_A \otimes \one_A} &=
                                                                     \sum_{i=1}^k
                                                                     \paren{\one_A
                                                                     \otimes \one_{A'}}\paren{(e_i \otimes \one)(\one \otimes e_i)^T
              + (\one \otimes e_i)(e_i \otimes \one)^T}\paren{\one_A
                                                                     \otimes
                                                                     \one_A}
  \\
 &= \sum_{i=1}^k \paren{\one_A^Te_i} \frac{k}{2} \frac{k}{2}
  \paren{e_i^T\one_A} +
  \frac{k}{2}\paren{\one_{A'}^Te_i}\paren{e_i^T\one_A}\frac{k}{2} \\
  &= \frac{k^3}{8}.
\end{align*}
Thus we have that if $k$ is even, the bisection bandwidth is $k^{\ell+1}$.

We now turn to the case where $k$ is odd. Because of the parity issues
in this case, it will be convient to define the bipartition
inductively.  To that end, let $(B,\overline{B})$ be a 
bipartition of $C(G,\ell-2)$ which witnesses the bandwidth such that
$\size{B} +1 = \size{\overline{B}}$.  Now define the sets
\begin{align*}
  Y &= [k]^{\ell-2} \times \paren{ (A \times A') \cup (A' \times A) 
    \cup (\set{k} \times [k-1]) } \cup \paren{ B \times 
      \set{k} \times \set{k}} \\
      \overline{Y} &= [k]^{\ell-2} \times \paren{ (A \times A) \cup (A' \times A') 
    \cup ([k-1] \times \set{k})} \cup \paren{ \overline{B} \times 
          \set{k} \times \set{k}}.
\end{align*}
It is clear that since $\size{A} = \size{A'}$ and $\abs{\size{B} -
  \size{B'}} = 1$, that $(Y,\overline{Y})$ is a bipartition of
$[k]^{\ell}$.  Abusing notation slightly, and we denote the set $[k]^{\ell-2} \times \paren{ (A \times A')
  \cup (A' \times A)}$ by $X$ and the set $[k]^{\ell-2} \times \paren{(A
  \times A) \cup (A' \times A')}$ as $\overline{X}$.  If we again let
$A_{\ell}$ denote the adjacency matrix of $C(G,\ell)$, we have that
\[ \one_X^T A_{\ell} \one_{\overline{X}} = (k-1)^3k^{\ell-2} \] by similar arguments as
above.  Additionally, we note that we have that
\begin{align*}
\paren{\one \otimes \one_A}^T M \paren{\one \otimes \one_{[k-1]}} &=
                                                                    k(k-1)^2
  \\
  \paren{\one \otimes \one_A}^T M \paren{\one \otimes e_k} &= k(k-1)
  \\
  \paren{\one \otimes \one_{[k-1]}}^T M \paren{\one \otimes e_k} &=
                                                                   2k(k-1)
  \\
  \paren{\one_A \otimes \one_{A'}}^T M \paren{e_k \otimes
  \one_{[k-1]}} & = \frac{(k-1)^2}{4}
  \\
  \paren{\one_A \otimes \one_{A'}}^T M \paren{e_k \otimes e_k} & = 0
  \\
    \paren{\one_{[k-1]} \otimes e_k}^T M \paren{e_k \otimes e_k} & = k-1. 
  \\
    \paren{\one_{[k-1]} \otimes e_k}^T M \paren{e_k \otimes
  \one_{[k-1]}} & = k(k-1)
  \\
\end{align*}
Putting these calculations together, we get that the bandwidth of the
partition $(Y,\overline{Y})$ is
\[  (k-1)k^{\ell} 
  + \one_{[k]^{\ell-2} \times[k-1] }^T
  \paren{I_{k^{\ell-3}}\otimes M}\one_{B
    \times \set{k}}+\one_{B
    \times \set{k} \times \set{k}}^T A_{\ell}\one_{B
    \times \set{k} \times \set{k}}.\]
Observing that $A_{\ell} = A_{\ell-2}\otimes I_{k^2} + I_{k^{\ell-3}}
  \otimes M \otimes I_k + I_{k^{\ell-2}} \otimes M$, it is easy to see
  that
  \[ \one_{B
    \times \set{k} \times \set{k}}^T A_{\ell}\one_{\overline{B}
    \times \set{k} \times \set{k}}= \BW(C(G,\ell-2)) +\one_{B
    \times \set{k}}^T  \paren{I_{k^{\ell-3}}
  \otimes M} \one_{\overline{B}
  \times \set{k} } \]
Now we note that terms involving $\one_{B \times \set{k}}$
sum to
\begin{align*}
 \paren{\one_{[k]^{\ell-1}} - \one_{B
                      \times \set{k}}}^T
  \paren{I_{k^{\ell-3}}\otimes M}\one_{B
    \times \set{k}} &= \one_{[k]^{\ell-1}}
                      \paren{I_{k^{\ell-3}}\otimes M}
                      \one_{B\times \set{k}}  -
                      \one_{B
                      \times \set{k}}^T
  \paren{I_{k^{\ell-3}}\otimes M}\one_{B\times \set{k}} \\
  &= 
                      2k\size{B} -
                      \one_{B
                      \times \set{k}}^T
    \paren{I_{k^{\ell-3}}\otimes M}\one_{B \times \set{k}} \\
                    &\leq
                      k\paren{k^{\ell-2}-1}
\end{align*}
                    Thus we have that
                    \begin{align*}
                      \BW\paren{C(G,\ell)} &\leq
                      (k-1)k^{\ell} + k^{\ell-1} 
                                             - k +
                                             \BW\paren{C(G,\ell-2)}  \\
                      &= k^{\ell+1}  -  k^{\ell} + k^{\ell-1} -  k + 
                        BW\paren{C(G,\ell-2)}.
                    \end{align*}
 Now, as $\BW\paren{C(G,1)} \leq k^2$ and
 $\BW\paren{C(G,2)} \leq k^3$, it is easy to see that by induction
 $\BW\paren{C(G,\ell)} \leq k^{\ell+1}$.
\qed \end{proof}

\subsubsection{$G$-connected-$H$}\label{SS:GconnH}
The $G$-connected-$H$ construction generalizes several different
constructions,  such as the Peterson Torus and Dragonfly topologies
discussed in this section as well as the Cube Connected Cycle topology
discussed in Section \ref{SS:CCC}.    To see this, we first formally
define what we mean by a $G$-connected-$H$ topology.

\begin{definition}[$k$-fold $G$-connected-$H$, $G \conn_k H$]
 A $k$-fold $G$-connected-$H$ topology, $\mathcal{G} = G \conn_k H$, is constructed from a $d$-regular
$G$ and a $r$-regular $td$-vertex graph $H$.  The
 vertex set of $\mathcal{G}$ is $V_G \times V_H$ and $G[ \set{g}
 \times V_H]$ is isomorphic to $H$ for all vertices $g \in V_G$.  The
 remaining edges form a $k$-regular graph on $V_G \times V_H$ satisfying that
 \[
e\!\paren{ \set{v} \times V_H, \set{v'} \times V_H} = \begin{cases} kt &
  \set{v,v'} \in E_G \\ 0 & \textrm{otherwise}\end{cases}.\]
When $k=1$ will surpress the subscript and simply write $G \conn H$.
\end{definition}

Oftentimes, $G$ is a Cayley graph and so the $k$-regular graph on $V_G \times
V_H$ can be defined by a mapping from the generators of $G$ to ordered
pairs in $V_H^2$.  For example, we denote by $Q_k$ the $k$-dimensional
hypercube, we can view the Cube Connecte Cycle topology of Section
\ref{SS:CCC} as a $1$-fold $Q_k$-connected-$C_k$.  More concretely, we note
that $Q_k$ can be represented as the Cayley graph on $\Z_2^k$
generated by the standard basis vectors, $\set{e_1,\ldots, e_k}$.
Since the generators of $Q_k$ have order two, the matching edges can
be formed by associating each generator with a fixed vertex of $C_k$.

This viewpoint can be extend to more complicated topologies, such as
the Peterson Torus~\cite{JungHyun2008}.
\begin{definition}[Peterson Torus, $\PT(a,b)$]
Let $a, b \geq 2$ such that at least one of $a$ or $b$ is odd.  Define
the vertex set of 
the Peterson Torus Topology, $\PT(a,b)$, as the set of ordered triples $(x,y,p)$
where $0 \leq x < a$, $0 \leq y < b$, and $0 \leq p < 10$.  Fixing the labels of the Peterson graph as given in Figure \ref{F:Peterson} the edge relationship is defined as:
\begin{itemize}
      \item \textbf{internal edge} $\paren{x,y,p}$ is adjacent to $\paren{x,y,q}$ if
      $p$ and $q$ are adjacent in the Petersen graph.
      \item \textbf{longitudinal edge} $\paren{x,y,6}$ is adjacent to $\paren{x,y+1,9}$.
      \item \textbf{latitudinal edge} $\paren{x,y,1}$ is adjacent to
      $\paren{x+1,y,4}$.
      \item \textbf{diagonal edge} $\paren{x,y,2}$ is adjacent to
      $\paren{x+1,y+1,3}$.
      \item \textbf{reverse diagonal edge} $\paren{x,y,7}$ is adjacent to
      $\paren{x-1,y+1,8}$.
      \item \textbf{diameter edge} $\paren{x,y,0}$ is adjacent to
      $\paren{x+\floor{\nicefrac{a}{2}},y+\floor{\nicefrac{b}{2}},5}$.
    \end{itemize}
  \end{definition}
This can be seen as a $1$-fold $G$-connected-$H$ graph where $G$ being the
Cayley graph on $Z_a \times \Z_b$ with generator set $\set{ \pm (0,1),
  \pm (1,0) , \pm (1,1), \pm (-1,1), \pm
  (\floor{\nicefrac{a}{2}},\floor{\nicefrac{b}{2}})}$  and $H$ being
the Peterson graph. We note that the condition that one of $a$ or $b$
is odd, is simply to ensure that the generator  $(\floor{\nicefrac{a}{2}},\floor{\nicefrac{b}{2}})$ is not it's own inverse and so $G$ has degree 10.  By allowing multiple edges in $G$, this restriction can be eliminated.  

\begin{figure}
\centering
\subfloat[Peterson Graph]{\label{F:Peterson}
      \begin{tikzpicture}
        \node[] (0) at (2,5) {0};
        \node[] (1) at (4,3.5) {1};
        \node[] (2) at (4,0) {2};
        \node[] (3) at (0,0) {3};
        \node[] (4) at (0,3.5) {4};
        \node[] (5) at (2,3.5) {5};
        \node[] (6) at (2.9,1.25) {6};
        \node[] (7) at (1,2.5) {7};
        \node[] (8) at (3,2.5) {8};
        \node[] (9) at (1.1,1.25) {9};
        \draw[-] (0.south) -- (5.north);
        \draw[-] (0.south west) -- (4.north east);
        \draw[-] (0.south east) -- (1.north west);
        \draw[-] (4.south) -- (3.north);
        \draw[-] (4.south east) -- (7.north west);
        \draw[-] (5.south) -- (9.north east);
        \draw[-] (5.south) -- (6.north west);
        \draw[-] (1.south west) -- (8.north east);
        \draw[-] (1.south) -- (2.north);
        \draw[-] (7.east) -- (8.west);
        \draw[-] (7.east) -- (6.north west);
        \draw[-] (8.west) -- (9.north east);
        \draw[-] (9.south west) -- (3.north east);
        \draw[-] (6.south east) -- (2.north west);
        \draw[-] (3.east) -- (2.west);
      \end{tikzpicture}
}
\hspace{1in}
\subfloat[Reduced $\PT(a,b)$]{\label{F:PetersonTorus}
      \begin{tikzpicture}
        \node[] (0) at (2,5) {0};
        \node[] (1) at (4,3.5) {1};
        \node[] (2) at (4,0) {2};
        \node[] (3) at (0,0) {3};
        \node[] (4) at (0,3.5) {4};
        \node[] (5) at (2,3.5) {5};
        \node[] (6) at (2.9,1.25) {6};
        \node[] (7) at (1,2.5) {7};
        \node[] (8) at (3,2.5) {8};
        \node[] (9) at (1.1,1.25) {9};
        \draw[-] (0.south) -- (5.north);
        \draw[-] (0.south west) -- (4.north east);
        \draw[-] (0.south east) -- (1.north west);
        \draw[-] (4.south) -- (3.north);
        \draw[-] (4.south east) -- (7.north west);
        \draw[-] (5.south) -- (9.north east);
        \draw[-] (5.south) -- (6.north west);
        \draw[-] (1.south west) -- (8.north east);
        \draw[-] (1.south) -- (2.north);
        \draw[-] (7.east) -- (8.west);
        \draw[-] (7.east) -- (6.north west);
        \draw[-] (8.west) -- (9.north east);
        \draw[-] (9.south west) -- (3.north east);
        \draw[-] (6.south east) -- (2.north west);
        \draw[-] (3.east) -- (2.west);
        \draw[red, bend right] (0.south) to (5.north);
        \draw[red, bend left] (4.east) to (1.west);
        \draw[red, bend left] (7.east) to (8.west);
        \draw[red, bend right] (9.east) to (6.west);
        \draw[red, bend left] (3.east) to (2.west);
      \end{tikzpicture}
}
\caption{Peterson Graph and the Peterson Graph with the Symmetric Function induced by the Peterson Torus}
\end{figure}
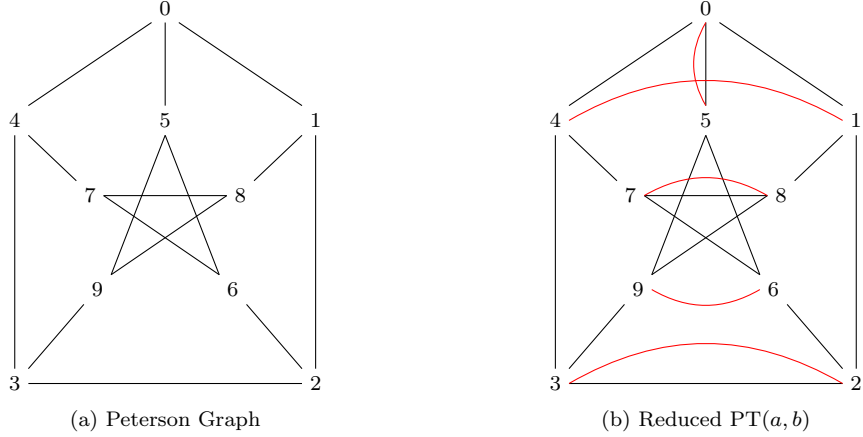

We first consider the bisection bandwidth of $G \conn_k H$.  As this will depend explicitly
on the number of nodes and edges in $H$, it is helpful to recall some
standard notation first.  Following the notation of West~\cite{West:GraphTheory}, the
number of nodes in a network $G$ will be denoted by $\size{G}$ and the
number of edges will be denoted $\order{G}$.

\begin{proposition}\label{L:GconnH_BW}
Let $\mathcal{G} = G \conn_k H$, then the bisection bandwidth of
$\mathcal{G}$ is at most $\frac{\size{G}\size{H}}{2\order{G}}k \BW(G)
+ \BW(H)$.
\end{proposition}

\begin{proof}
We note that if $\size{G}$ is even, then the bipartition of $G$
yielding $\BW(G)$, lifts naturally to a bipartition of $G \conn_k H$.
As each edge in $G$ is represented by
$\frac{\size{G}\size{H}}{2\order{G}}k$ edges in $G \conn_k H$, this gives an upper
bound of $\frac{\size{G}\size{H}}{2\order{G}}k \BW(G)$.   If instead,
$\size{G}$ is odd, the natural lift of the minimal bipartition doesn't
yield a bipartition of $G \conn_k H$.  However, this can be corrected
by splitting one of the copies of $H$, yielding the extra $\BW(H)$
term.  
  \qed \end{proof}

We now turn the algebraic connectivity of $G \conn_k H$.  Because of the general structure of the matching edges and the
potentially unstructured nature of $G$ and $H$, the reduction lemma
can not be applied in general to $G$-connected-$H$ graphs.  However,
there is still a natural symmetry formed by the $G$-connected-$H$
structure, specifically the identification of vertices by common $G$
labels or common $H$ labels.  However, because of the lack of
automorphism structure we must turn to eigenvalue interlacing results
such as the following by Haemmers.

\begin{lemma}{\cite[Corollary 1.8]{Brualdi:GraphsandMatrices}}\label{L:Brualdi}
Let $A$ be an $n \times n$ real-symmetric matrix with eigenvalues
$\lambda_1 \geq \lambda_2 \geq \cdots \geq \lambda_n$.  Let $\alpha_1,\ldots,
\alpha_m$ be a partition of the integers $\set{1,2,\ldots,n}$ into $m$
nonempty consecutive sets of integers, where $\size{\alpha_i} = n_i$.
Let $A_{ij}$ be the submatrix of $A$ defined by the entries whose row
is in $\alpha_i$ and column is in $\alpha_j$.  Define $B$ as the $m\times m$
real symmetric matrix with \[ b_{ij} =
  \frac{\one_{n_i}^TA_{ij}\one_{n_j}}{n_i}.\]  The eigenvalues of $B$
interlace the eigenvalues of $A$, in particular, $\lambda_2(B) \leq
\lambda_2$.  
  \end{lemma}

  \begin{proposition}\label{L:GconnH}
Let $G$ be a connected $d$-regular graph and let $H$ be a
connected $r$-regular, $td$-vertex graph, and let $\mathcal{G} = G
\conn_k H$ be a $k$-fold  $G$-connected-$H$ graph.  Let $\lambda_2$ be the
second largest eigenvalue of $G$, then the algebraic connectivity of
$\mathcal{G}$ is at most $k - \frac{k\lambda_2}{d}$.  
\end{proposition}

\begin{proof}
Let $A$ be the adjacency matrix of $\mathcal{G} = G \conn_k H$.  We will
proceed to show that $\lambda_2(\mathcal{G}) \geq
\frac{1}{d}\lambda_2+ r$ and then the desired result follows
immediately from the $(r+k)$-regularity of $\mathcal{G}$.  To this end
we will apply Lemma \ref{L:Brualdi} to the partion of the vertices
given by $\set{ \set{v} \times V_H}_{v \in V_G}$.    Abusing notation,
for any $v, v' \in V_G$ we will denote by $A_{vv'}$
the submatrix induced by the rows $\set{v} \times V_H$ and columns
$\set{v'} \times V_H$.  Noting that $\size{ \set{v} \times V_H} = td$
for all $v \in V_G$, we have that
\[
\one_{td}^T A_{vv'} \one_{td} = \begin{cases} rtd & v = v' \\ tk &
  \set{v,v'} \in E_G \\  0 & \set{v,v'} \not\in E_G \end{cases}
\]
and thus $B = rI + \frac{k}{d}A_G$ where $A_G$ is the adjacency matrix of
the graph $G$. The interlacing of the eigenvalues of $B$ and
$\mathcal{G}$ provides the result immediately.
  \qed \end{proof}

The strong dependence on the spectrum of $G$ is unsurprising as the $G$-connected-$H$ graphs
implicitly inherent the connectivity structure of $G$, while increasing the relative degrees in a way that
doesn't improve the spectral behavior of $G$.  In particular, we note
that another way of deriving Lemma \ref{L:GconnH} is to apply the
Raleigh-Ritz formulation of $\lambda_2(\mathcal{G})$ and use the
vector $\frac{1}{\sqrt{td}}\one \otimes w_2$ where $(\lambda_2,w_2)$
is the second largest eigenpair of $A_G$.

It is natural to consider the implications of Lemma \ref{L:Brualdi}
when partitioning on the $H$-coordinate instead of the
$G$-coordinate.  Unfortunately, because of the unstructured nature of
the $k$-regular graph relatively little can be said.  However, if the graph $G$
is Cayley graph and the matching edges are tied to the generator set
then the automorphisms of $G$ (specifically, those that follow from
vertex transitivity of Cayley graphs) imply that there is an
automorphism of $G \conn H$ such that the orbits are given by $V_G
\times \set{h}$ for $h \in V_H$.  Then the Reduction Lemma yields
that there is multi-graph whose spectrum is a subset of the
spectrum of $G \conn_k H$.  Specifically, taking the graph $H$ plus a
$k$-regular graph (allowing self-loops) coming from the structure of
the $k$-regular graph in $G \conn_k H$.  As an example, the Peterson
Torus can be reduced with the reduction lemma to the graph illustrated
with in Figure \ref{F:PetersonTorus} with the red edges corresponding
to the matching edges.  Computing the algebraic connectivity of the
reduced graph yields that the $\rho_2$ for the Peterson torus is at
most 2.  While this is small, it can be reduced further by applying
Proposition \ref{L:GconnH}.

\begin{corollary}
Let $a \geq b \geq 2$ such that at least one of $a$ or $b$ is odd.
The algebraic connectivity of $\PT(a,b)$ is at most $\frac{4 -
  3\cos\paren{\frac{4\pi}{a}} - \cos\paren{\frac{2\pi}{a}}}{5}$ and the
bisection bandwidth is at most $6b + ab + 5$.  
  \end{corollary}

  \begin{proof}
By Proposition \ref{L:GconnH} to bound the algebraic connectivity it
suffices to find the second largest eigenvalue of the Cayley graph $G$ on
the group $\Gamma = \Z_a \times \Z_b$ generated by
\[S = \set{ \pm (1,0), \pm (0,1),
  \pm (1,1), \pm (1,-1), \pm \paren{
    \floor{\nicefrac{a}{2}},\floor{\nicefrac{b}{2}}}}.\]   Let $\chi
\colon \Gamma \rightarrow \C^{\size{\Gamma}}$ be the character table
for $\Gamma$.  We recall that the spectrum of $G$ is explicitly given
by the multiset
\[ \set{\sum_{s \in S} \chi_g(s) \mid g \in \Gamma}, \]
see for instance \cite{Brouwer:SpectraGraphs}.  As $\Gamma$ is the product
of two cyclic groups, it is straightforward to explicitly determine the
character table and get that the spectrum is given by the multiset
\[ \set{ \sum_{(s,t) \in S} e^{\frac{2\pi i x}{a} s} e^{\frac{2 \pi i
        y}{b} t} \mid (x,y) \in \Z_a \times Z_b}.\]
In particular, this gives that $\lambda_2$ for $G$ is given by
\[ 2\max_{ (x,y) \not\equiv
      (0,0)}
\cos\!\paren{\frac{2\pi x}{a}} + \cos\!\paren{\frac{2 \pi y}{b}} +
2\cos\paren{\frac{2\pi x}{a}}\cos\!\paren{ \frac{2\pi y}{b}} + \cos\paren{\frac{2\pi\floor{\nicefrac{a}{2}}x}{a} + \frac{2\pi\floor{\nicefrac{b}{2}}y}{b}}.\] 
It is relatively straightforward to see that the maximum is achieved
when $(x,y) = (2,0)$, yielding that $\lambda_2$ is at least
\[ 2+6\cos\paren{\frac{4\pi}{a}} + 2
  \cos\paren{\frac{4\pi\floor{\frac{a}{2}}}{a}} \geq 2 +
  6\cos\paren{\frac{4\pi}{a}} + 2\cos\paren{\frac{2\pi}{a}}.\]

The upper bound on the bisection bandwidth will follow from Lemma
\ref{L:GconnH_BW}.  Specifically, as the Peterson torus is a $G \conn
H$ with $G$ being the Cayley graph on $\Z_a \times \Z_b$ with
generators $S$ and $H$ the Peterson graph, the bisection bandwidth is
upper bounded by $\BW(G) + \BW(H)$.  As the girth of the Peterson
graph is 5, any collection of 5 vertices induces at most 5 edges.
Thus there are at least 5 edges crossing the cut, and this lower bound
is achieved exactly by taking any of the 5-cycles in the Peterson
graph.

For the bisection bandwidth of the graph on $\Z_a \times
\Z_b$, we will denote the vertices by $[a] \times [b]$.  If $a$ is even, then
the set $T = [\frac{a}{2}] \times [b]$ induces a bipartition with $6b +
ab$ edges crossing the cut, that is, the edges corresponding to
elements $\set{\frac{a}{2}} \times [b]$ and the generators
$\set{(1,-1),(1,0),(11)}$, the edges corresponding to the elements
$\set{1} \times [b]$ and the generators $\set{(-1,-1),(-1,0),(-1,1)}$, and the edges corresponding to an arbitrary
vertex of $T$ and the generators $\set{
 \paren{\floor{\frac{a}{2}},\floor{\frac{b}{2}}},
 \paren{-\floor{\frac{a}{2}},-\floor{\frac{b}{2}}}}$.  In the case
that $a$ is odd, we consider the set $[\floor{\frac{a}{2}}] \times [b]
\cup \set{\ceil{\frac{a}{2}}} \times [\floor{\frac{b}{2}}]$ and in a
similar manner get that there are $6b + 2\floor{\frac{ab}{2}}$ edges
crossing the cut, completing the proof.
\qed \end{proof}

\subsubsection{DragonFly}
As we will see, the DragonFly topology will end up being a specific
class of $G$-connected-$H$ topologies and can be understood in terms
of the results of Section \ref{SS:GconnH}, however, due to their
recent importance in ``readily'' available supercomputing topologies~\cite{Cray:Slingshot,Cray:XC}
we address them separately in this section.  The motivating idea
behind the DragonFly topology is to maximize the performance of a
supercomputing topology while minimizing the overall cost of the
system.  To that end, Kim, Dally, Scott, and Abts designed the
DragonFly topology around a two-level hierarchy~\cite{Kim2008}.  The
top level network employs an optical network to
communicate over long distances (i.e.\ across the physical layout of
the supercomputer), while the second layer employs an electrical
network to communicate short distances (i.e.\ intrarack communication)
and reduce the overall cost.   While the specifications of Kim, et
al.\ allow for arbitrary topologies for both the optical and
electrical portions of the topology, the typically implementation uses
a fully-connected optical network combined with some other network for
the electrical network, oftentimes either fully-connected or
a Butterfly variant.  For example, the Cray Slingshot interconnect
(which is being used for NSERC's Perlmutter system) uses 64
port switches to build a DragonFly topology based on all-to-all
connections for both the optical and electrical networks.
\begin{definition}[DragonFly, $\DF(H)$]
  If $H$ is an $n$-vertex, $r$-regular graph, then the DragonFly
  topology with parameter $H$ consists of $n+1$ copies of $H$ together
  with a matching such that each edge goes between distinct copies of
  $H$.  Alternatively, $\DF(H)$ may be thought of as a $1$-fold $K_n
  \conn H$.  
\end{definition}

We note that since the DragonFly topology can be represented as $G
\conn H$ topology we immediately have bounds on the algebraic
connectivity and bisection bandwidth. 

\begin{corollary}
Let $H$ be a connected graph and let $D$ be the DragonFly topology
generated by $H$.  The algebraic connectivity of $D$ is at most $1 +
\frac{\size{H}}{2\order{H}}$ and the bisection bandwidth is at most
\[ \paren{\frac{\size{H}+1}{2}}^2 + \BW(H).\]
  \end{corollary}

  \begin{proof}
Noting that $D = K_{\size{H}+1} \conn H$ and the second largest
adjacency eigenvalue of the complete graph is $-1$, the bound on the
algebraic connectivity follows immediately from Proposition
\ref{L:GconnH}.  To provide the upper bound on the bisection
bandwidth, consider a equipartition of the $\size{H}+1$ copies of
$H$.  If $\size{H}$ is odd, then the only edges crossing the partition
are ``matching'' or ``optical'' edges and there are
$\paren{\frac{\size{H}+1}{2}}^2$ of them.  However, if $\size{H}$ is
even, then one of the copies of $H$ must also be partitioned yielding
\[ \paren{\frac{\size{H}}{2}}^2 + \frac{\size{H}}{2} + \BW(H) \leq
  \paren{\frac{\size{H}+1}{2}}^2 + \BW(H)\] edges crossing the
partition.  
    \qed \end{proof}

\subsubsection{SlimFly} In \cite{Besta:SlimFly}, Besta and Hoefler suggested that it
would be advantageous to consider topologies that have close to the
maximum number of nodes for a given radix and diameter.  The upper
bound on the number of nodes of a $k$-regular graph of diameter $d$ is
given by $1 + k\sum_{i=0}^{d-1}(k-1)^i$ and is referred to as the
Moore bound.  The class of graphs exactly achieving this bound, known
as Moore graphs, has been extensively studied and shown to have
significant limitation on both the radix and size, see \cite{Miller:MooreSurvey}.

In this context, Besta and Hoefler propose the SlimFly topology based on the
construction of McKay, Miller, and \v{S}ir\'{a}n~\cite{McKay:MMS} which is close to achieving the
Moore bound.  These SlimFly topologies have a single parameter $q$,
which is a prime power such that $q \equiv 1 \pmod{4}$ and results in
a topology on $2q^2$ nodes with degree $\frac{3q-1}{2}$.

\begin{definition}[SlimFly, $\SF(q)$] Let $\zeta$ be a primitive
  $q^{\textrm{th}}$-root of unity over the Galois field $\mathbb{F}_q$.  The vertices are then elements of
$\set{0,1} \times \mathbb{F}_q \times \mathbb{F}_q$.  The edge set is
broken into three sets:
\begin{enumerate}
\item $\set{(0,x,y),(0,x,y')}$ where $y - y' = \zeta^i$ and $i \equiv
  0 \pmod{2}$,
\item $\set{(1,m,c),(1,m,c')}$ where $c-c' = \zeta^j$ and $j \equiv 1
  \pmod{2}$, and
\item $\set{(0,x,y),(1,m,c)}$ where $y = mx + c$.
\end{enumerate}
\end{definition}

\begin{proposition}\label{L:slimfly_alpha}
Let $q$ be a prime-power such that $q \equiv 1 \pmod{4}$.  The
algebraic connectivity of the SlimFly topology with parameter $q$ is 
$q$.
\end{proposition}

\begin{proof}
In order to bound the algebraic connectivity, we will use the
Reduction Lemma.  To that end, let $\zeta$ be a primitive root of the
Galois field $\mathbb{F}_q$ and define $\gamma_{\zeta}$ by
$(0,x,y) \mapsto (0,x,y + \zeta)$ and $(1,m,c) \mapsto
(1,m,c+\zeta)$.  It is easy to see that this is an automorphism of the
SlimFly topology and that the orbits of the group generated by this automorphism are given by
$\set{0}\times\set{x}\times \mathbb{F}_q$ and $\set{1}\times \set{m}
\times\mathbb{F}_q$ for $x,m \in \mathbb{F}_q$.  As an arbitrary
element $(0,x,y) \in \set{0}\times \mathbb{F}_q \times \mathbb{F}_q$,
has precisely one neighbor in the orbit $\set{1} \times \set{m} \times
\mathbb{F}_q$ for any $m \in \mathbb{F}_q$, namely $(1,m,y-mx)$, we
  have that the reduction graph $H$ is a complete bipartite graph
  $K_{q,q}$ with $\frac{q-1}{2}$ self-loops at every vertex.  As the
  algebraic connectivity of this graph is $q$, by the Reduction Lemma
  we have that the algebraic connectivity of the is at most $q$.  

Now we will show that the algebraic connectivity is exactly $q$.  To
this end, recall that the eigenspace associated to any eigenvalue that is not present in the
spectrum of the reduced graph has the property that the entries sum to
zeros over all of the orbits.  That is, if $v$ is such an eigenvector
and $\one_{\sigma}$ is the indicator function for the orbit $\sigma$,
then $v^T\one_{\sigma} = 0$.  Furthermore, since the orbits of the
automorphism are Cayley graphs on $\mathbb{F}_q$, the eigenvectors
can be expressed
in terms of the characters of $\paren{\mathbb{F}_q,+}$.
Additionally, the eigenvalues associated to $\chi_f$ are given by the
character sums over the generators.  Specifically, the eigenvalue
associated to the non-trivial character $\chi$ on the Cayley graph generated by $\zeta^{2j}$ is $\mu =
\sum_{j=1}^{\nicefrac{q-1}{2}} \chi(\zeta^{2j})$, while the
eigenvalue associated to $\chi$ on the Cayley graph generated by
$\zeta^{2j-1}$ is $\sum_{j=1}^{\nicefrac{q-1}{2}} \chi(\zeta^{2j-1})
= -1 - \mu$.  Thus let $\set{\chi_f}_{f = 1}^{q-1}$ be the set of
non-trivial characters of $\paren{\mathbb{F}_q,+}$.  The eigenvector
$v$ can then be expressed as 
\[ \sum_{i \in \set{0,1}} \sum_{x \in \mathbb{F}_q} \sum_{f \in
    [q-1]} \frac{a_{i,x,f}}{\sqrt{q}} e_i \otimes e_x \otimes
  \chi_f \]
where $\sum_{i,x,f} \alpha_{i,x,f}^2 = 1$

Now letting $A$ be the adjacency matrix of the SlimFly topology, we
consider the quadratic form $\overline{v}^TAv$ in three parts.  The
portion corresponding to edges induced by $\set{0}\times \mathbb{F}_q
\times \mathbb{F}_q$, the portion corresponding
to edges induced by $\set{1}\times \mathbb{F}_q
\times \mathbb{F}_q$, and the portion corresponding to edges between
these two sets.  It is easy to see that the contribution of the edges
internal to these two sets are given by $\sum_{f = 1}^{q-1} \sum_{x \in
  \mathbb{F}_q} \overline{\alpha_{0,x,f}}\alpha_{0,f,x}
\mu_f$ and $\sum_{f = 1}^{q-1} \sum_{m \in
  \mathbb{F}_q} \overline{\alpha_{0,m,f}}\alpha_{1,m,f}
(-1-\mu_f)$, respectively.  Recalling the edges between the two
sets are governed by the relationship $y = mx +c$ for $(x,y), (m,c)
\in \mathbb{F}_q^2$, we have that the contribution of those edges to
the quadratic form is
\[ \sum_{f,g = 1}^{q-1} \sum_{(x,y) \in \mathbb{F}_q^2} \sum_{(m,c)
    \in \mathbb{F}_q^2}
  \frac{\overline{\alpha_{0,x,f}}\alpha_{1,m,g}}{q} \overline{\chi_f(y)}\chi_g(c)
  \one_{y = mx+c} + \frac{\alpha_{0,x,f}\overline{\alpha_{1,m,g}}}{q} \chi_f(y)\overline{\chi_g(c)}
  \one_{y = mx+c} .\]  
Now we note that the non-zero entries in the sum occur when $c =
y-mx$.  Furthermore, $\chi_g$ is a homomorphism into $(\C,\times)$ so
$\chi_g(y-mx) = \chi_g(y)\chi_g(-mx)$.  By additionally recalling that
$\set{\chi_f}_{f=1}^{q-1}$ is an orthogonal basis, this sum simplifies
to
\[ \sum_{f=1}^{q-1} \sum_{x \in \mathbb{F}_q} \sum_{y \in
    \mathbb{F}_q} \sum_{m \in \mathbb{F}_q}
  \overline{\alpha_{0,x,f} \vphantom{ ()}}\alpha_{1,m,f}\chi_f(-mx) + 
\overline{\alpha_{1,m,f} \vphantom{
    ()}}\alpha_{0,x,f}\overline{\chi_f(-mx)}.\]
Thus, letting $M$ be the diagonal matrix formed from
$\set{\mu_f}_{f=1}^{q-1}$, we have that the norm of the quadratic form
is bounded above by the largest eigenvalue of 
\[\mathcal{M} =  \left[ \begin{matrix} M & I \\ I & -M -I\end{matrix}\right].\]  
Motivated by this formulation we consider the auxiliary problem
\[ \max_{x^2+y^2 = 1} \mu x^2 + 2xy - (\mu+1)y^2.\]   
Noting that we may assume that $x,y \geq 0$, this can be
reparameterized as 
\[ \max_{\delta \in [-1,1]} \mu\paren{\frac{1}{2} +
    \frac{1}{2}\delta} + 2 \sqrt{\frac{1}{2} +
    \frac{1}{2}\delta}\sqrt{\frac{1}{2} - \frac{1}{2}\delta}
  - \paren{\mu+1}\paren{\frac{1}{2} - \frac{1}{2}\delta}.\]
The derivative of the objective function is $\frac{2\mu+1}{2} -
\frac{\delta}{\sqrt{1-\delta^2}}$ with roots $\pm
\frac{2\mu+1}{\sqrt{\paren{2\mu+1}^2 +4}}$.  Thus the largest eigenvalue
of $\mathcal{M}$ is 
\[\max_{f \in [q-1]}\max\set{ \mu_f, -\mu_f-1, -\frac{1}{2} +
    \frac{1}{2} \sqrt{\paren{2\mu_f+1}^2 + 4}, -\frac{1}{2} - \frac{1}{2}
  \frac{\paren{2\mu_f+1}^2- 4}{\sqrt{\paren{2\mu_f+1}^2+4}}}.\]
Using the fact that the Cayley graph generated by the odd powers of
$\zeta$ and the Cayley graph generated by the even powers of $\zeta$
are isomorphic (via $x \mapsto \zeta x$), this reduces to
$-\frac{1}{2} + \frac{1}{2}\sqrt{\paren{2\mu+1}^2 + 4}$ where $\mu$ is
the second largest eigenvalue of the Cayley graph generated by the
even powers of $\zeta$.  Using the fact that this Cayley graph is edge
transitive and has diameter 2,  we get that $\mu \leq \frac{q-1}{2} -
\frac{1}{4}\frac{q-1}{2}$ (see \cite[Section 7.3]{Chung:spectral}).
Combining these results we have that the largest eigenvalue not
represented in the reduced graph is at most 
\[ -\frac{1}{2} + \frac{1}{2} \sqrt{\paren{\frac{3}{4}q +
      \frac{1}{4}}^2 + 4} < \frac{q-1}{2}\] for $q \geq 5$.
\qed \end{proof}

\begin{proposition}
Let $q$ be a prime-power such that $q \equiv 1 \pmod{4}$.  The
bisection bandwidth of the SlimFly topology with parameter $q$ is at
most 
$\frac{q(q^2+1)}{2}$ and at least $\frac{q^3}{2}$.
\end{proposition}

\begin{proof}
Let $X \subset \mathbb{F}_q$ such that $\left| X \right| =
\frac{q-1}{2}$ and let $\overline{X}$ be the complement of $X$.  We
consider the bipartition $\set{0} \times X \times \mathbb{F}_q
  \cup \set{1} \times \overline{X} \times \mathbb{F}_q$.  We note that
  there are no edges between $\set{0} \times X \times \mathbb{F}_q$
  and $\set{0} \times \overline{X} \times \mathbb{F}_q$ and similarly there
  are no edges between $\set{1} \times X \times \mathbb{F}_q$
  and $\set{1} \times \overline{X} \times \mathbb{F}_q$.  Now, as
  $\set{0}\times\set{x}\times \mathbb{F}_q$ has exactly one edge to
  $\set{0} \times \set{m} \times \mathbb{F}_q$ for every $x,m \in
  \mathbb{F}_q$.  Thus the bisection bandwidth of the SlimFly
  topology is at most $q \paren{\frac{q-1}{2}}^2 +
  q \paren{\frac{q+1}{2}}^2 = \frac{q(q^2+1)}{2}$.  

The lower bound follows from Lemma \ref{L:slimfly_alpha} and the lower
bound on the bandwidth based on the algebraic connectivity.
\qed \end{proof}

It is worth mentioning that the gap between the bisection bandwidth achieved by a $\frac{3q-1}{2}$ regular graph on $2q^2$ vertices and the bisection bandwidth of the SlimFly topology could be attributed to fact that the SlimFly topology is not a Moore graph.  In fact, it is straightforward to construct a bisection of a Moore graph whose bisection bandwidth asymptotically matches the known lower bounds on the bisection bandwidth of a similar Ramanujan graph.  

\begin{proposition}
Let $G$ be a Moore graph with regularity $q$ and girth $2d+1$.  The bisection bandwidth of $G$ is at most $\frac{q}{2} + \frac{q^2}{4}\paren{q-1}^{d-1}$ if $q$ is even and $q + \frac{q^2-1}{4}\paren{q-1}^{d-1}$ if $q$ is odd.
\end{proposition}

\begin{proof}
Fix an arbitrary vertex $v$ in $G$ and let its neighbors be $w_1, \ldots, w_q$.  Since the girth of Moore graph is $2d+1$, the diameter is $d$.  For $i \in [q]$ define $W_i$ as the set of vertices whose shortest path to $v$ goes through $w_i$ and define $S_i \subset W_i$ as the vertices are at distance precisely $d$ from $v$.  Note that since $G$ is a Moore graph, for any vertex $s \in S_i$ all the neighbors of $s$ must be in distinct sets $S_j$ where $j \neq i$. 

Suppose first that $q$ is even and consider the bipartition $\paren{ \paren{\cup_{i=1}^{\frac{q}{2} W_i} }\cup \set{v}, \cup_{i=\frac{q}{2}+1}^q W_i}$.   Now clearly each edge in each of the $W_i'$'s does not cross the bipartition, and so the only edges we need concern ourselves with are those adjacent to $v$ and those adjacent to vertices of $S_i$.  Now as each vertex in $S_i$ is adjacent to a vertex in each of the $S_j$'s except $S_i$, this implies that there are $\frac{q}{2} + \frac{q}{2} \sum_{i=1}^{\frac{q}{2}} \size{S_i} = \frac{q}{2}\paren{1+\frac{q}{2}\paren{q-1}^{d-1}}$ edges crossing the bisection.  

The construction for $q$ odd is similar to the one for $q$ even, except rather than placing all of $W_{\frac{q+1}{2}}$ on one side of the partition, the partitioning procedures is done of the trees rooted at the vertices of distance 2 from $v$ in $W_{\frac{q+1}{2}}$.  
\qed \end{proof}


\section{Conclusion}

\begin{figure}
\centering
\includegraphics[width=\textwidth]{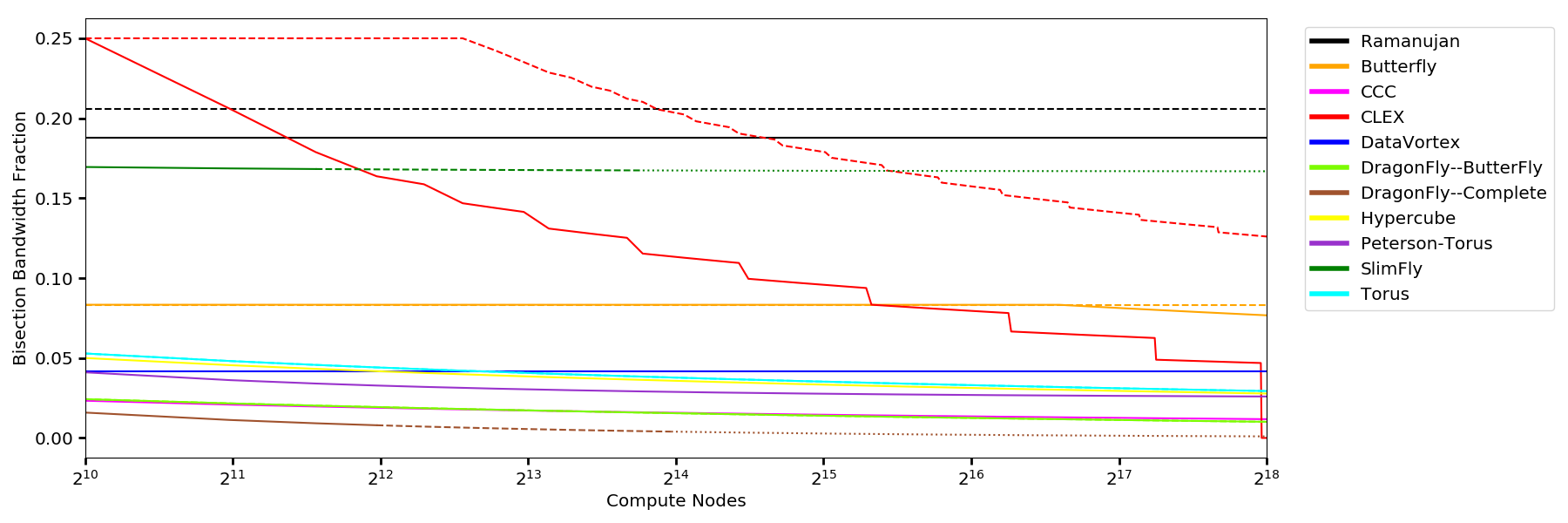}
\caption{Proportional Bisection Bandwidth for supercomputing topologies by number of compute nodes.} \label{fig:dens}
\end{figure}

We provide in Table \ref{T:boundsBW} a summary of the results on the
bisection bandwidth and algebraic connectivity of the topologies
considered in this work. Additionally, for comparison we provide
bounds on the bisection bandwidth and algebraic connectivity for a
similarly sized Ramanujan topology.  We focus on bisection bandwidth
in our comparison, although we remind the reader the spectral results
summarized in Table \ref{T:boundsBW} also provide bounds on a plethora
of other salient interconnection network properties (such as diameter,
average distance, and fault tolerance) via the theorems mentioned in
Section \ref{sec:back}. As closer inspection of the table makes clear,
for each of these topologies there is a significant gap between the
achieved value and the minimum guaranteed to be achievable in an
equivalent Ramanujan topology. However, assessing these results {\it
  across} families is more challenging due to different input
parameters and parameter multiplicities for each topology. To better
enable such a comparison, in Figure \ref{fig:dens} we plot the
proportional\footnote{Relative to sum of the graph degrees, or twice
  the number of links} bisection bandwidth by number of compute nodes
for each topology, as well as the minimum guaranteed by a Ramanujan
topology.  In general the solid lines represent those topologies
with switches comparable to current topologies (that is, having radix
at most 64 as in the Cray Slingshot Topology~\cite{Cray:Slingshot,Cray:XC} while the dashed
lines represent the proportional bisection bandwidth achievable with
next generation switches (radix at most 128), and the dotted line
represents those topologies that would require even higher radix
switches.  We note that even the limitations on the radix are not
sufficient to uniquely determine the highest bisection bandwidth
proportion for some topologies.  Thus we will also impose following
additional assumptions on the topologies with an aim of avoiding
trivial instantiations of the topology:
\begin{itemize}
\item \textbf{Butterfly:} for the Butterfly topology we assume that
  there are at least 3 ranks of switches, i.e.\ $s \geq 3$,

\item \textbf{CLEX:}  for the CLEX topology we assume that there are
  at least two layers $\ell \geq 2$ and that the initial generating
  graph is the complete graph on at least 3 vertices,

\item \textbf{Data Vortex:} for the Data Vortex we assume that there
  are at least 3 ``cylinders'', i.e.\ $C \geq 3$,

\item \textbf{DragonFly -- Butterfly:} similarly to the Butterfly
  topology, for the DragonFly topology where the electrical network is
  given by a Butterfly network, we assume that $s \geq 3$, and

  \item \textbf{Torus:} for the torus topology we assume that all the
    cycles are non-degnerate, i.e. that $k \geq 3$.
  \end{itemize}
Even as we compare these upper bounds on the best-possible bisection bandwidth for each topology against the worst-possible in a Ramanujan topology, we still observe a sizable gap, with the 128 radix SlimFly and CLEX topologies the closest to the Ramanujan lower bound. We suspect the region where CLEX outperforms Ramanujan graphs is an artifact of the looseness of the analysis of CLEX for small parameter settings, rather than a true reflection of the relative sizes of the bisection bandwidths.

\afterpage{
  \newgeometry{left=3cm,right=3cm,top=3cm,bottom=3cm}
  \begin{landscape}
  \vspace*{\fill}\mbox{}
   \begin{table}[h]
     \centering
\renewcommand{\arraystretch}{2}
\begin{tabular}{l | c | c | c |c  || c|   c}
& nodes & radix &  $\rho_2$ upper bound &$\BW$ upper bound  & Ramanujan $\rho_2$ & Ramanujan $\BW$  \\ \hline
Butterfly$(s,k)$ & $sk^s$ & $2k$ & $\frac{4\pi^2
                                   k}{n^2}$&
                                                    $\frac{\paren{k+1}k^{s}}{2}$ & $2k - 2\sqrt{2k-1}$& $\frac{s k^{s+1}}{2} - \bigOh{sk^{s+\nicefrac{1}{2}}}$ \strut\\[2pt] \hline
CCC$(d)$& $d2^d$ & 3 &  $\frac{\pi^2}{\paren{d+2}^2}$& $2^{d-1}$&
                                                                  $3-2\sqrt{2}$
                                                                                 &$\frac{3-2\sqrt{2}}{4}
                                                                                   d2^d\paren{1
                                                                                   +
                                                                                   \lilOh{1}}$ \strut\\[2pt] \hline
CLEX$(k,\ell)$&$k^{\ell}$&$2\ell k -k -1$&$4k-2$&$
                                                                 k^{\ell+1}$&
                                                                               $\paren{1-\lilOh{1}}\paren{2\ell k - 2\sqrt{2\ell k}}$\strut & $ \frac{\ell k^{\ell+1}}{2} - \bigOh{k^{\ell}\sqrt{k\ell}}$ \strut  \\[2pt] \hline
$\DV(A)$&$AC2^{C-1}$&4&$A2^{C-1}$\strut&$A2^{C-2}$\strut
                                    & $ 4-2\sqrt{3}$ &$(1-\nicefrac{\sqrt{3}}{2})AC2^{C-1}\paren{1+\lilOh{1}}$\\[2pt] \hline
$\DF(H)$ & $\size{H}^2 + \size{H}$ & $\frac{2\order{H}}{\size{H}}+1$ & $1 +
                                                \frac{1}{\size{H}}$
                            &                                                               $
                                                               \paren{\frac{\size{H}+1}{2}}^2
                         + \BW\paren{H}$ &   $\frac{2\order{H}}{\size{H}}+1-2\sqrt{\frac{2\order{H}}{\size{H}}}$
  \strut&
                                                  $\frac{\paren{\size{H}+1}\paren{2\order{H}+\size{H}}}{4}
                                                  - \bigOh{\sqrt{\size{H}^3\order{H}}}$
                                                                 \strut
  \\[2.5pt] \hline
  $G \conn_k H$ & $\size{G}\size{H}$ & $\frac{2\order{H}}{\size{H}}+k$ &
                            $k -
                                                                       \frac{k\size{G}}{2\order{G}}\lambda_2(G)$&
                                                                                                                          $\frac{k\size{G}\size{H}}{2\order{G}}\BW(G)+
                                 \BW(H)$ &
                                           $\frac{2\order{H}}{\size{H}}+k-2\sqrt{\frac{2\order{H}}{\size{H}}

                                           + k -1}$
                                                         &
                                                           $\frac{\order{G}\paren{2\size{H}
                                                           +
                                                           k\order{H}}}{4}
  - \bigOh{\size{G}\sqrt{\order{H}\size{H}}}$\\[2pt] \hline
  Hypercube$(d)$ & $2^d$ & $d$ &$2$& $2^{d-1}$ &$d-2\sqrt{d-1}$ &
                                                                  $d2^{d-2}-\bigOh{\sqrt{d2^d}}$  \strut \\[2pt] \hline
  $\PT(a,b), a \geq b$ & $10ab$ & $4$& $\frac{4 -
                                       3\cos\paren{\frac{4\pi}{a}} -
                                       \cos\paren{\frac{2\pi}{a}}}{5}$
                            & $6b + ab + 5$ & $4-2\sqrt{3}$ &
                                               $ab\paren{10-5\sqrt{3}}\paren{1+\lilOh{1}}$ \\[2pt] \hline
  SlimFly$(q)$ & $2q^2$ & $\frac{3q-1}{2} $&$q$& $\frac{q^3+q}{2}$ &$\frac{3q-1}{2} - \sqrt{6q-6} $ & $\frac{3q^3}{4}- \bigOh{q^{\nicefrac{5}{2}}}$ \strut \\[2pt]\hline 
Torus$(k,d)$& $k^d$ & $2d$ &$2(1-\cos(\frac{2 \pi}{k}))$ \strut& $2k^{d-1}$ &$2d-\sqrt{8d-4}$ & $dk^{d-1}-\bigOh{\sqrt{dk^d}}$\strut \\[2pt] \hline 
\end{tabular}
\caption{Table of bounds on algebraic connectivity and bisection bandwidth for common supercomputing topologies}\label{T:boundsBW}
\end{table}
\vspace*{\fill}
\end{landscape}
\restoregeometry}

In light of the beneficial structural properties of random graphs, it is natural to ask whether any potential utility of Ramanujan supercomputing topologies is already offered by randomized constructions, such as the well-known Jellyfish topology. Indeed, such topologies are touted for their low diameter, short average path lengths, and high bisection bandwidth \cite{singla2012jellyfish}. Although random regular graphs are not quite Ramanujan, it is true that random $d$-regular graphs have good spectral expansion. Notably, Friedman's celebrated proof \cite{Friedman2003} of Alon's second eigenvalue conjecture \cite{Alon:EigenvaluesExpanders} showed that if $G$ is a random $k$-regular graph on $n$ vertices then with probability going to 1 as $n \to \infty$, we have $\lambda(G) \leq 2\sqrt{k-1}+o(1)$.  Thus, in the limiting sense, random regular graphs are ``almost Ramanujan." Nonetheless, randomized constructions are also limited as interconnection topologies in that they pose serious challenges for routing, physical layout, and wiring \cite{singla2012jellyfish}. In these regards, structured topologies offer advantages.

Consequently, one may ask whether more structured families, such as Cayley graphs, might serve as a more amenable alternative to random constructions.  Since many of the popular topologies can be phrased as Cayley graphs (e.g. the torus and hypercube topologies) or have a strong connections to Cayley graphs (e.g. the SlimFly and Peterson torus topologies) it is natural to speculate that a Cayley graph could serve as the basis of a strong supercomputing topology. Indeed, work \cite{Akers1987a} investigating Cayley graphs as interconnection networks dates back to at least the 1980's, see \cite{Heydemann1997} for a survey. In particular, {\it abelian} Cayley graphs may seem particularly promising because the classification of abelian groups gives a natural means of easily performing efficient routing.  However, abelian Cayley graphs do not offer the spectral expansion of Ramanujan graphs: as a consequence of a result of Cioab\u{a}~\cite{Cioaba:CayleySpectra} there is a constant $C(k,\epsilon)$ such that if the group has more elements than $C(k,\epsilon)$, then any Cayley graph generated by a $k$-element set has algebraic connectivity at most $\epsilon$.  Thus, for any fixed radix $k$, there does not exist an infinite family of radix $k$ abelian Cayley graphs which are Ramanujan.

Given these tradeoffs between randomized designs and highly structured Cayley graph designs, we believe the explicit Ramanujan construction by Lubotsky, Phillips, and Sarnak warrants further investigation as a candidate for supercomputing interconnection networks. By virtue of their optimal spectral expansion, LPS graphs offer many of the same (if not better) structural properties exhibited by random regular graphs. Yet, as highly structured Cayley graphs, LPS graphs may be more amenable to practical considerations and easier to develop efficient routing schemes for than random constructions. Indeed, recent work by Sardari \cite{Sardari2017}, as well as Pinto and Petit \cite{Pinto2018} investigating short paths in LPS graphs shows that, while sometimes challenging to analyze, the local structure of these topologies may be exploited for the purposes of routing. While the work we've done here attests to the structural benefits of LPS graphs over other supercomputing topologies, additional work is needed to better assess the benefits of utilizing LPS graphs as interconnection networks in practice.

\begin{acknowledgements}
We would like to thank Andres Marquez, Kevin Barker, and Carlos Ortiz-Marrero for helpful discussions. This work was supported by the High Performance Data Analytics program at PNNL. Information Release PNNL-SA-147472. 
\end{acknowledgements}

\bibliographystyle{siam}
\bibliography{scRefs}

\end{document}